\documentclass{article}

\usepackage{tikz}
\usepackage{amsfonts}
\usepackage{amssymb,epic,eepic}
\usepackage{amsmath}
\usepackage{dcolumn}
\usepackage{bm}
\usepackage{bbm}
\usepackage{verbatim}
\usepackage{color}
\usepackage{amsthm}

\usetikzlibrary{arrows,positioning,decorations.pathmorphing}

\newcommand{\hr}{{\mathcal H}}
\newcommand{\cn}{{\mathcal N }}

\newcommand{\cs}{{\mathcal S}}
\newcommand{\crr}{{\mathcal R}}
\newcommand{\fr}{{\mathcal F}}

\newcommand{\fri}{{\mathfrak I}}
\newcommand{\kr}{{\mathcal K}}

\newcommand{\cc}{{\mathbb C}}
\newcommand{\rr}{{\mathbb R}}
\newcommand{\M}{{\mathcal M}}
\newcommand{\nn}{{\mathbb N}}

\newcommand{\eps}{{\varepsilon}}        

\newcommand{\A}{\mathcal A}
\newcommand{\B}{\mathcal B}

\newcommand{\cP}{\mathcal P}
\newcommand{\bS}{\mathbf S}
\newcommand{\bX}{\mathbf X}
\newcommand{\bY}{\mathbf Y}
\newcommand{\eins}{{\mathbbm{1}}}
\newcommand{\BIGOP}[1]
{
\mathop{\mathchoice%
{\raise-0.22em\hbox{\Large $#1$}}%
{\raise-0.05em\hbox{\large $#1$}}{\hbox{\large $#1$}}{#1}}}
\newcommand{\bigtimes}{\BIGOP{\times}}
\newcommand{\BIGboxplus}{\mathop{\mathchoice%
{\raise-0.35em\hbox{\huge $\boxplus$}}%
{\raise-0.15em\hbox{\Large $\boxplus$}}{\hbox{\large $\boxplus$}}{\boxplus}}}

\newtheorem{theorem}{Theorem}

\newtheorem{corollary}{Corollary}

\newtheorem{definition}{Definition}

\newtheorem{lemma}{Lemma}

\newtheorem{remark}{Remark}

\newcommand{\tr}{\mathrm{tr}}
\newcommand{\supp}{\mathrm{supp}}

\DeclareMathOperator{\conv}{conv}
\DeclareMathOperator{\aff}{aff}
\DeclareMathOperator{\ri}{ri}

\begin{document}

\title{Arbitrarily Small Amounts of Correlation for Arbitrarily Varying Quantum Channels}
\author{H. Boche, J. N\"otzel \\
\scriptsize{Electronic addresses: \{boche, janis.noetzel\}@tum.de}
\vspace{0.2cm}\\
{\footnotesize Lehrstuhl f\"ur Theoretische Informationstechnik, Technische Universit\"at M\"unchen,}\\
{\footnotesize 80290 M\"unchen, Germany}
}
\maketitle

\begin{abstract}
As our main result we show that, in order to achieve the randomness assisted message - and entanglement transmission capacities of a finite arbitrarily varying quantum channel it is not necessary that sender and receiver share (asymptotically perfect) common randomness. Rather, it is sufficient that they each have access to an unlimited amount of uses of one part of a correlated bipartite source. This access might be restricted to an arbitrary small (nonzero) fraction per channel use, without changing the main result.\\
We investigate the notion of common randomness. It turns out that this is a very costly resource - generically, it cannot be obtained just by local processing of a bipartite source. This result underlines the importance of our main result.\\
Also, the asymptotic equivalence of the maximal- and average error criterion for classical message transmission over finite arbitrarily varying quantum channels is proven.\\
At last, we prove a simplified symmetrizability condition for finite arbitrarily varying quantum channels.
\end{abstract}
\begin{section}{Introduction and Historical Remarks}
An arbitrarily varying quantum channel (we will use the shorthand 'AVQC' henceforth) is defined by a set $\fri=\{\cn_s\}_{s\in\bS}$ of quantum channels which all have the same input- and output system. We will consider the case $|\bS|<\infty$ only. The techniques necessary to cope with arbitrary sets have been developed in \cite{abbn}. They are based on the case $|\bS|<\infty$. The in- and output systems are controlled by a sender and a receiver, who (in the cases considered here) wish to transmit either entanglement or classical messages. We do not consider strong subspace transmission here, since the (asymptotical) equivalence to entanglement transmission (over arbitrarily varying quantum channels) has already been proven in \cite{abbn}.\\
They do so by $l$-fold usage of the arbitrarily varying channel, which is under the control of a third party. This third party, called the adversary, is able to select either one of the channels $\cn_{s^l}:=\cn_{s_1}\otimes\ldots\otimes\cn_{s_l}$ for which $s^l\in\bS^l$. It is understood that sender and receiver have to select their protocol first, after that the adversary makes his choice of channel sequence $s^l$. As usual in the (quantum) Shannon information theoretic setup we are interested in the case when $l$ goes to infinity.\\
This scenario can also be understood as an attack on the communication between a legitimate sender and receiver. In the present case, the only aim of the adversary is to prohibit the communication, no eavesdropping is done. The strength of the adversary can easily be modeled by the set $\fri$ from which he is allowed to select the channels. Clearly, in case that $\{\cn_s\}_{s\in\bS}=\mathcal C(\hr,\kr)$, (the set of all completely positive and trace preserving maps with input system $\hr$ and output system $\kr$) any kind of communication between sender and receiver is impossible.\\
Recent work \cite{abbn} provided a formula for the entanglement transmission capacity of such a channel when sender and receiver are allowed to use an unlimited amount of shared randomness in order to perform a possibly correlated randomization over their encoding and decoding strategies. The corresponding capacity was named the 'random entanglement transmission capacity', or $\A_{\mathrm{random}}(\fri)$ for short. The authors of \cite{abbn} also showed that it is sufficient to use only a polynomial (in the number of channel uses) amount of \emph{common randomness} to achieve (rates that are close to) $\A_{\mathrm{random}}(\fri)$.\\
Using this result, they then showed that it was sufficient for sender and receiver to be able to establish this common randomness first by sending classical messages, and then use a bunch of \emph{deterministic} entanglement transmission codes afterwards. This led to their 'Quantum Ahlswede Dichotomy', which stated that the deterministic entanglement transmission capacity $\A_{\mathrm{det}}(\fri)$ of an AVQC $\fri$ equals its random entanglement transmission capacity, if its deterministic message transmission capacity $\overline C_{\mathrm{det}}(\fri)$ is greater than zero. Since $\A_{\mathrm{det}}(\fri)\leq\overline C_{\mathrm{det}}(\fri)$, the very same statement holds true with $\overline C_{\mathrm{det}}(\fri)$ replaced by $\A_{\mathrm{det}}(\fri)$.\\
In the conclusions of their paper they conjectured that $\A_{\mathrm{det}}(\fri)=\A_{\mathrm{random}}(\fri)$ holds for every AVQC $\fri$.\\
This question is important from two points of view. First, for transmission of messages over classical arbitrarily varying channels, there exist explicit examples (see \cite{ahlswede-elimination}) where the capacity without at least a small, polynomial in the number of channel uses, amount of shared randomness between sender and receiver is zero. However, if one uses a small amount of randomness, it is strictly larger than zero.\\
For message transmission over arbitrarily varying classical-quantum channels, these examples have been extended in \cite{ahlswede-blinovsky}. So, for this setting as well, a small amount of shared randomness can boost the message transmission capacity from zero to some value strictly larger than zero.\\
A comparable behaviour is widely known in the information theoretic context from complexity theory, where the use or abandonment of randomized algorithms defines the different complexity classes $P$ and $NP$. The interest in such questions is not of purely academic nature - randomness in any form is an additional resource, that can become costly - and this is exactly the second reason, why above conjecture is important:\\
Consider a whole network of senders and receivers that want to carry out different quantum communication tasks such as entanglement transmission, generation, distillation, message transmission, identification of states or messages and so on. Given that the channels between the mutual participants of the network are never perfectly shielded against the environment or the actions of the other participants, one can ask how the network would benefit from the distribution of shared randomness or some other sort of stabilizing resource, like e.g. correlation, over the network.\\
In this paper we prove that certain elementary networks (AVcqCs) strongly benefit already from a very cheap form of shared randomness, namely correlation.\\
Correlation arises for example when sender and receiver both agree to watch correlated events in order to synchronize their en- and decoding. The signals of one of the many satellites orbiting earth could be used for this matter, as well as observations of weather or other natural processes. Important here is not that sender and receiver receive the \emph{same} signal, but only that it is \emph{correlated}: The temperature at a given time, measured at two nearby places, will usually not be the same, but it is also not statistically independent. What we would like to point out is that, although the use of a satellite might seem more useful for synchronization, much weaker resources could be used as well!\\
In the present paper, we thus not only start an investigation on the conjectured equality $\A_{\mathrm{det}}(\fri)=\A_{\mathrm{random}}(\fri)$ - we do so by searching for the least amount of randomness that is, to our current knowledge, sufficient for transmission of entanglement at $\A_{\mathrm{random}}(\fri)$. The approach that is used produces, at the same time, similar results for message transmission.\\
It is our hope that, with this approach, we might finally be able to show that, for entanglement transmission over arbitrarily varying quantum channels, \emph{no} sort of randomness is helpful at all. To start with, we will now distinguish between four different types of entanglement- and message transmission codes for AVQCs. Each class requires a different kind of resource. We will introduce them now, and then explain what they are capable of.
\\\\
\textbf{1) Deterministic codes.} These are the ones that require the least resources. Sender and receiver simply agree on one encoding- and decoding scheme, the adversary selects a channel sequence, and the transmission starts. A transmission of entanglement or messages at a certain rate is successfull, if it is asymptotically perfect for \emph{every} choice of (infinite) channel sequence that the adversary might come up with. The corresponding capacities for message -and entanglement transmission include the term 'deterministic' in their name.
\\\\
\textbf{2) $((X,Y),r)-$correlated codes.} A bipartite source, modelled by an i.i.d. random variable $(X,Y)$ with values in some finite set $\bX\times\bY$, is observed by sender and receiver. The sender has access to the random variable $X$ and the receiver to $Y$. Every $r$-th channel use, they obtain one pair of realizations of $(X,Y)$. This way, they are able to make their choice of encoding and decoding \emph{dependent} on the outcomes of $(X,Y)$. The adversary, while aware of the description of $(X,Y)$, does not have access to the specific outcomes. In order to avoid trivialities, we will always assume that the mutual information between $X$ and $Y$ is larger than zero: $I(X,Y)>0$. Any such pair $((X,Y),r)$ will also be named 'correlation'.
\\\\
\textbf{3) Common randomness assisted codes.} Sender and receiver both have access to the outcomes of a source that, for $l$-fold usage of the channel, outputs one pair of elements taken from a set $\Gamma_l\times\Gamma_l$. It is guaranteed that the probability distribution according to which these elements are chosen converges to the equidistribution on the subset of pairs of identical elements and that $|\Gamma_l|$ grows unbounded with $l$. If such a source is available, we will equivalently speak of 'common randomness'.
\\\\
\textbf{4) Random codes.} This is the most general class. It consists of the whole set of probability measures on the set of encoding and decoding schemes, where the sigma-algebra may be chosen in any way that makes the error criteria integrable functions. It contains all the other classes as special cases.
\\\\
Our results, in order of appearance and put into historical context, are the following.
\\\\
First, common randomness is a stronger resource than mere correlation. Whenever the source $(X,Y)$ belongs to the relative interior of the set of probability distributions on $\bX\times\bY$ (this is equivalent to $p(x,y)>0$ for all $(x,y)\in\bX\times\bY$ for the distribution $p$ that $(X,Y)$ is distributed according to), then if sender and receiver cannot communicate with each other, \emph{no} common randomness can be extracted from $(X,Y)$. That not even a finite amount of common randomness (meaning that $(|\Gamma_l|)_{l\in\nn}$ converges to some constant) can be extracted already follows from the results of \cite{witsenhausen}.\\
On the contrary, it is evident that a large enough amount of common randomness (meaning that $\liminf\frac{1}{l}\log|\Gamma_l|$ is large enough) allows the sender and receiver to asymptotically simulate the statistics of \emph{any} sequence $((X,Y)^{\otimes l})_{l\in\nn}$, if local randomness is for free. We will take this assumption for granted, as is usually done (e.g. in wiretap scenarios).\\
Our results are in fact even slightly stronger. We show that the generation of a certain class of sequences of bipartite distributions that includes those who model common randomness is impossible under the above assumptions. The class is defined in such a way that every member of it satisfies the minimal requirements on a sequence of bipartite distributions that has the property that it enables classical message transmission over an AVQC, just below the corresponding random message transmission capacity $\overline C_{\mathrm{random}}(\fri)$, when the AVQC satisfies the symmetrizability condition.\\
The validity of this condition is both necessary and sufficient for an AVQC $\fri$ to satisfy $\overline C_{\mathrm{det}}(\fri)=0$, as was proven in \cite{abbn}. It is, to our knowledge, the first time that it is put to use. The complicated, non-single letter character of the symmetrizability condition does not really cry out for applications, making this latter statement nontrivial.\\
The question of quantifying the correlation present in a bipartite source $(X,Y)$ in an operational way has drawn research interest also into other directions, including the results of \cite{witsenhausen} that are used in our proof.\\
A different approach was taken in the work \cite{gacs-koerner}, whose authors were concerned with the question of finding an operational interpretation of the mutual information $I(X,Y)$ of a bipartite source. Their approach was to not only generate some common randomness, as we would like to, but also be able to reconstruct the local sources $X$ and $Y$ from it. Therefore, their results were of a negative character and, due to the additional task of reproducing the marginal statistics they do not find an application in our work.\\
Later on, in his work \cite{wyner}, Wyner discarded the earlier results \cite{witsenhausen} and \cite{gacs-koerner} by Witsenhausen and Gacs, K\"orner in his search for an operational notion for the 'common information' or 'common core' of a bipartite i.i.d. source. He found different mathematical formulations for these words and could characterize them in terms of an (probably one-shot, but that seems not to be immediate) entropic formula involving the source. Again, these results are not directly applicable to our scenario.\\
The authors of \cite{ahlswede-koerner}, who provided a broader overview of the topic, concluded that there were even more possible notions of common information, making the definitions and results of \cite{wyner} just one of many.\\
Recently, the question of how to quantify the correlation that is present in a bipartite source $(X,Y)$ has been studied by Kang and Ulukus in \cite{kang-ulukus}. They defined yet another measure of correlation. They showed that their measure was non-increasing in the number of i.i.d. copies of $(X,Y)$ and connected them to the results of \cite{witsenhausen} and \cite{gacs-koerner}.\\
Their results were generalized to the quantum setting by Beigi in \cite{beigi}.\\
Another approach to the generation of common randomness in our sense is given in a scenario where sender and receiver not only have access to their respective parts of $(X,Y)$ but may also send classical messages between each other. This case is considered for example in \cite{ahlswede-csiszar-I} and \cite{ahlswede-csiszar-II}, a corresponding quantum case has been dealt with in \cite{devetak-winter}.
\\\\
Second we show that, despite the above remarks, $((X,Y),r)-$correlated codes are already enough to achieve either $\A_{\mathrm{random}}(\fri)$ or $\overline C_{\mathrm{andom}}(\fri)$. This is regardless of the value of $r$. Our proof follows the reasoning of Ahlswede and Cai in \cite{ahlswede-cai}, who showed the equivalent result for arbitrarily varying classical channels.\\
Like in the proof of our first result, we explicitly use the symmetrizability conditions given in \cite{abbn}, making this the second time that they are successfully put to use.\\
This result is not only an interesting step in the investigation of the question whether correlated and randomized coding schemes help the transmission of entanglement over AVQCs, it is also relevant from a very practical point of view.\\
If it turns out that there are in fact cases where correlated, randomized coding schemes for message- or entanglement transmission have a strictly better performance than the respective deterministic variants, then the better performance can already be reached by using only $((X,Y),r)$ correlation. But this kind of correlation is rather cheap compared to common randomness (by our first result) and, additionally, it might be gained by simply observing e.g. some natural process.\\
Also, even though we conjecture that the equality $\A_{\mathrm{random}}(\fri)=\A_{\mathrm{det}}(\fri)$ holds for all arbitrarily varying quantum channels $\fri$, there is an explicit example in \cite{ahlswede-blinovsky} where $\overline C_{\mathrm{random}}(\mathbb W)>0$ and $\overline C_{\mathrm{det}}(\mathbb W)=0$ holds for arbitrarily varying classical-quantum (cq) channels. In this case, by our Corollary \ref{corollary-2} below, it is already clear that mere correlation leads to a huge benefit for message transmission.\\
We conjecture that similar examples of AVQCs can be found, that is: an AVQC $\fri$ satisfying both $\overline C_{\mathrm{random}}(\fri)>0$ and $\overline C_{\mathrm{det}}(\fri)=0$.
\\\\
As a third result, we show that the deterministic message transmission capacity of an AVQC is independent of the choice of either maximal- or average error criterion. This fact is due to our choice of encoding- and decoding strategies. In the quantum case, it seems of little interest to restrict to pure state inputs or POVMs whose elements are given by orthogonal projections. But already in the classical scenario, the results of \cite{ahlswede-elimination} show this kind of behaviour of the respective capacities for randomized encoding strategies.
\\\\
Our fourth result is motivated by our conjecture concerning the classical message transmission capacities of an AVQC. It aims at finding nontrivial examples of AVQCs $\fri$ for which $\overline C_{\mathrm{det}}(\fri)=0$, but $\overline C_{\mathrm{random}}(\fri)>0$ holds. We felt little urge to prove a coding theorem for message transmission using random codes over an AVQC, but expect a regularized version with an additional optimization over input ensembles in the spirit of the coding result given in \cite{ahlswede-blinovsky}. Given such a result, it should in principle be possible to find out whether $\overline C_{\mathrm{random}}(\fri)>0$ holds, by brute force computer simulations at the worst.\\
The symmetrizability conditions given in \cite{abbn} cannot be checked that easily, especially not since it might be unstable with respect to small variations of the set $\fri$. By the results of Ahlswede, Bjelakovic, Boche and N\"otzel in \cite{abbn}, symmetrizability of an AVQC is equivalent to $\overline C_{\mathrm{det}}(\fri)=0$, hence a simpler version of the symmmetrizability conditions seems desirable. This simplified version is the fourth result of this paper.
\\\\
The paper is structured as follows: First, in Section \ref{sec:Notation} we fix the notation. We proceed with elementary definitions in Section \ref{sec:Definitions}. With these preparations at hand, we state our main results in Section \ref{sec:Main Results}. The rest of the paper (Sections \ref{sec:common-randomness} to \ref{sec:equivalence}) is devoted to the corresponding proofs, in the same order the results were stated and with a separate section for each proof.
\end{section}
\begin{section}{\label{sec:Notation}Notation}
All Hilbert spaces are assumed to have finite dimension and are over the field $\cc$. The set of linear operators from $\hr$ to $\hr$ is denoted $\mathcal B(\hr)$. The adjoint of $b\in\mathcal B(\hr)$ is marked by a star and written $b^\ast$.\\
$\cs(\hr)$ is the set of states, i.e. positive semi-definite operators with trace (the trace function on $\mathbb B(\hr)$ is written $\tr$) $1$ acting on the Hilbert space $\hr$. Pure states are given by projections onto one-dimensional subspaces. A vector $x\in\hr$ of unit length spanning such a subspace will therefore be referred to as a state vector, the corresponding state will be written
$|x\rangle\langle x|$. For a finite set $\mathbf X$ the notation $\mathfrak{P}(\mathbf X)$ is reserved for the set of probability distributions on $\mathbf X$, and
$|\mathbf X|$ denotes its cardinality. For any $l\in\nn$, we define $\bX^l:=\{(x_1,\ldots,x_l):x_i\in\bX\ \forall i\in\{1,\ldots,l\}\}$, we also write $x^l$ for the elements of $\bX^l$.\\
The set of completely positive trace preserving (CPTP) maps (also called quantum channels)
between the operator spaces $\mathcal{B}(\hr)$ and $\mathcal{B}(\kr)$ is denoted by $\mathcal{C}(\hr,\kr)$.\\
Closely related is the set of classical-quantum channels (abbreviated here using the term 'cq-channels') with finite input alphabet $\mathbf Z$ and output alphabet $\kr$, that arises from $\mathcal C(\hr,\kr)$ by setting $d=|\mathbf Z|$ and restricting the inputs to matrices that are diagonal in any specific basis. This set is denoted $CQ(\mathbf Z,\kr)$.\\
For any natural number $N$, we define $[N]$ to be the shorthand for the set $\{1,...,N\}$.\\
Using the usual operator ordering symbols $\leq$ and $\geq$ on $\mathcal B(\hr)$, the set of measurements with $N\in\nn$ different outcomes is written
\begin{align}\M_N(\hr):=\{\mathbf D:\mathbf D=(D_1,\ldots,D_N)\ \wedge\ \sum_{i=1}^ND_i\leq\eins_\hr\ \wedge\ D_i\geq0\ \forall i\in[N]\}.\end{align}
To every $\mathbf D\in \M_N(\hr)$ there corresponds a unique operator defined by $D_0:=\eins_\hr-\sum_{i=1}^ND_i$. Throughout the paper, we will assume that $D_0=0$ holds. This is possible in our scenario, since adding the element $D_0$ to any of the other $D_1,\ldots,D_N$ does not decrease the performance of a given code.\\
The von Neumann entropy of a state $\rho\in\mathcal{S}(\hr)$ is given by
\begin{equation}S(\rho):=-\textrm{tr}(\rho \log\rho),\end{equation}
where $\log(\cdot)$ denotes the base two logarithm which is used throughout the paper.\\
The Holevo information is for a given channel $W \in CQ(\mathbf{X},\hr)$ and input probability distribution $p \in \mathfrak P(\mathbf{X})$ defined by
\begin{align}
 \chi(p, W) := S(\overline{W}) - \sum_{x \in \mathbf{X}} p(x) S(W(x)),
\end{align}
where $\overline{W}$ is defined by $\overline{W} := \sum_{x \in \mathbf{X}} p(x) W(x)$\\
Given a bipartite random variable $(X,Y)$, its mutual information $I(X,Y)$ is given by $I(X,Y):=H(X)+H(Y)-H(X,Y)$, where $H(\cdot)$ is the usual Shannon entropy.\\
For $\rho\in\mathcal{S}(\hr)$ and $\cn\in
\mathcal{C}(\hr,\hr)$ the entanglement fidelity (which was defined in \cite{schumacher}) is given by
\begin{equation}F_e(\rho,\cn):=\langle\psi, (id_{\mathcal{B}(\hr)}\otimes \cn)(|\psi\rangle\langle \psi|)     \psi\rangle,  \end{equation}
with $\psi\in\hr\otimes \hr$ being an arbitrary purification of the state $\rho$.\\
For a finite set $\mathcal{W}=\{W_s\}_{s\in\bS} \subset \mathcal C(\hr,\kr)$ or $\mathcal{W}=\{W_s\}_{s\in\bS} \subset CQ(\mathbf Z,\kr)$ we denote its convex hull by $\conv(\mathcal{W})$ (for the  definition of the convex hull, \cite{webster} is a useful reference).
In the cases considered here the following will be sufficient. For a set $\mathcal{W}:= \{W_s\}_{s \in \bS}$
\begin{align}\label{eq:conv-hull}
 \conv(\mathcal{W})=\left\{W_{q}: W_q=\sum_{s\in \bS}q(s)W_s,\ q \in\mathfrak{P}(\bS)\right\}.
\end{align}
Finally, we need some simple topological notions for convex sets in finite dimensional normed space $(V, ||\cdot||)$ over the field of real or complex numbers which we borrow from \cite{webster}. Let $F\subset V$ be convex. $x\in F$ is said to be a relative interior point of $F$ if there is $r>0$ such that $B(x,r)\cap \aff F\subset F$. Here $B(x,r)$ denotes the open ball of radius $r$ with the center $x$ and $\aff F$ stands for the affine hull of $F$. The set of relative interior points of $F$ is called the relative interior of $F$ and is denoted by $\ri F$.
\end{section}
\begin{section}{\label{sec:Definitions}Definitions}
For the rest of this subsection, let $\fri=\{\cn_s\}_{s\in\bS}$ denote a finite AVQC. Henceforth, we follow the convention from \cite{abbn}, using the term 'the AVQC $\fri$' as a linguistic shorthand for the mathematical object $(\{\cn_{s^l}\}_{s^l\in\bS^l})_{l\in\nn}$.\\
We will now define the entanglement transmission capacities of an AVQC. Corresponding coding theorems can be found in \cite{abbn}.
\begin{definition}
An $(l,k_l)-$\emph{random entanglement transmission code} for $\fri$ is a probability measure $\mu_l$ on $(\mathcal C(\fr_l,\hr^{\otimes l})\times\mathcal C(\kr^{\otimes l},\fr_l'),\sigma_l)$,
where $\fr_l,\fr_l'$ are Hilbert spaces, $\dim\fr_l=k_l$, $\fr_l\subset\fr_l'$ and the sigma-algebra $\sigma_l$ is chosen such that the function $(\cP_l,\crr_l)\mapsto F_e(\pi_{\fr_l},\crr_l\circ\cn_{s^l}\circ\cP_l)$ is measurable
w.r.t. $\sigma_l$ for every $s^l\in\bS^l$.\\
Moreover, we assume that $\sigma_l$ contains all singleton sets. An example of such a sigma-algebra $\sigma_l$ is given by
the product of sigma-algebras of Borel sets induced on $\mathcal C(\fr_l,\hr) $ and $\mathcal C(\kr,\fr_l') $ by the standard topologies of the ambient spaces.
\end{definition}
\begin{definition}\label{def:random-cap-ent-trans}
A non-negative number $R$ is said to be an achievable entanglement transmission rate for the AVQC $\fri=\{\cn_s  \}_{s\in\bS}$ with random codes if there is a sequence of $(l,k_l)-$random entanglement transmission codes such that
\begin{enumerate}
\item $\liminf_{l\rightarrow\infty}\frac{1}{l}\log k_l\geq R$ and
\item $\lim_{l\rightarrow\infty}\inf_{s^l\in\bS^l}\int F_e(\pi_{\fr_l},\crr^l\circ\cn_{s^l}\circ\cP^l)d\mu_l(\cP^l,\crr^l)=1$.
\end{enumerate}
The random entanglement transmission capacity $\A_{\textup{random}}(\fri)$ of $\fri$ is defined by
\begin{equation}
\A_{\textup{random}}(\fri):=\sup\left\{R\in\rr_+:\begin{array}{l}R \textrm{ is an achievable entanglement trans-}\\ \textrm{mission rate for } \fri \textrm{ with random codes}\end{array}\right\}.
\end{equation}
\end{definition}
Now, we consider the important subclass of random entanglement transmission codes that use only \emph{correlation}:
\begin{definition}Let $(X,Y)$ be a bipartite random variable taking values in the finite alphabet $\bX\times\bY$ which is distributed according to some probability distribution $p\in\mathfrak P(\bX\times\bY)$.\\
For $l\in\nn$, $r\in\nn$ and $n(l):=\lfloor l/r\rfloor$, an $((X,Y),r)$ code for entanglement transmission over the finite AVQC $\fri$ is a random entanglement transmission code with $\sigma_l=\{\{ \cP_{x^{n(l)}} \}\}_{x^{n(l)}\in\bX^{n(l)}} \times \{\{ \crr_{y^{n(l)}} \}\}_{y^{n(l)}\in\bY^{n(l)}}$ being some finite subset of $\mathcal C(\fr_l,\hr^{\otimes l})\times\mathcal C(\kr^{\otimes l},\fr_l')$ and $\mu_l(\cP_{x^{n(l)}},\crr_{y^{n(l)}})=p^{\otimes n(l)}(x^{n(l)},y^{n(l)})$.\\
A nonnegative number $R$ is said to be an achievable $((X,Y),r)$ rate for entanglement transmission over $\fri$ if there is a sequence of $((X,Y,r)$ codes for entanglement transmission over $\fri$ such that
\begin{enumerate}
\item $\liminf_{l\rightarrow\infty}\frac{1}{l}\log k_l\geq R$ and
\item $\lim_{l\rightarrow\infty}\min_{s^l\in\bS^l}\sum_{x^{n(l)}\in\bX^{n(l)}}\sum_{y^{n(l)}\in\bY^{n(l)}} F_e(\pi_{\fr_l},\crr^l_{y^{n(l)}}\circ\cn_{s^l}\circ\cP^l_{x^{n(l)}})p^{\otimes n(l)}(x^{n(l)},y^{n(l)})=1$.
\end{enumerate}
The corresponding $((X,Y),r)$ entanglement transmission capacity of $\fri$ is defined through
\begin{align}
\A(\fri,r,(X,Y)):=\sup\left\{R\in\rr_+:\begin{array}{l}R \textrm{ is an achievable entanglement trans-}\\ \textrm{mission rate for } \fri \textrm{ using ((X,Y),r) codes}\end{array}\right\}.
\end{align}
\end{definition}
Having defined random codes and random code capacity for entanglement transmission we are in the position to introduce their deterministic counterparts: An $(l,k_l)-$code for entanglement transmission over $\fri$ is an $(l,k_l)-$random code for $\fri$ with $\mu_l(\{(\mathcal{P}^l,\crr^l)  \}  )=1$ for some encoder-decoder pair $(\mathcal{P}^l,\crr^l)$ (This explains our requirement on $\sigma_l$ to contain all singleton sets) and $\mu_l(A)=0$ for any $A\in\sigma_l$ with $(\mathcal{P}^l,\crr^l)\notin A $. We will refer to such measures as point measures in what follows.
\begin{definition}
A non-negative number $R$ is a deterministically achievable entanglement transmission rate for the AVQC $\fri=\{\cn_s  \}_{s\in \bS}$ if it is achievable in the sense of Definition \ref{def:random-cap-ent-trans} for random codes  with \emph{point measures} $\mu_l$.\\
The deterministic entanglement transmission capacity $\A_{\textup{det}}(\fri)$ of $\fri$ is given by
\begin{equation}
\A_{\textup{det}}(\fri):=\sup\left\{R\in\rr_+:\begin{array}{l}R \textrm{ is an achievable entanglement trans-}\\ \textrm{mission rate for } \fri \textrm{ with deterministic codes}\end{array}\right\}.
 \end{equation}
\end{definition}
Let us take a brief sidestep to the more restricted model (which is not only interesting in its own right, but also important for our current investigations on AVQCs) of an AVcqC that was introduced in \cite{ahlswede-blinovsky}. We are again given a finite subset $\mathbb W=\{W_s\}_{^s\in\bS}$ of channels, but this time $\mathbb W\subset CQ(\mathbf Z,\kr)$ for some finite alphabet $\mathbf Z$. Then $(\{W_{s^l}\}_{s^l\in\bS^l})_{l\in\nn}$ is abbreviated 'the AVcqC $\mathbb W$'. Again, codes, rates and capacities can be defined. We are going to make use of the following:
\begin{definition}\label{def:avc_rand-code}
An $(l,M_l)$-\emph{random code} for message transmission over $\mathbb W=\{W_s\}_{s\in \bS}$ is a probability measure $\mu_l$ on $(\mathbf
(X^l)^{M_l}\times\M_N(\hr^{\otimes l}),\Sigma_l)$. Again, $\Sigma_l$ is a sigma-algebra that contains the singleton sets. An example has been given in \cite{bbjn}. In order to shorten our notation, we write elements of $(X^l)^{M_l}\times\M_N(\hr^{\otimes l})$ in the form
$(x^l_i,D_i)_{i=1}^{M_l}$.
\end{definition}
\begin{definition}\label{def:avc_det-code}
An $(l,M_l)$-\emph{deterministic code} for message transmission over $\mathbb W=\{W_s\}_{s\in \bS}$ is given by a random code for message transmission over
$\mathbb W$ with $\mu_l$ assigning probability one to a singleton set.
\end{definition}
\begin{definition}\label{def:avc_rand-achievability}
A non-negative number $R$ is called achievable for transmission of messages over the AVcqC $\mathbb W=\{W_s\}_{s\in \bS}$ with random codes using the average error criterion if
there is a sequence $(\mu_l)_{l\in\nn}$ of $(l,M_l)$-random codes such that the following two lines are true:
\begin{equation}
\liminf_{l\to\infty}\frac{1}{l}\log M_l\ge R
\end{equation}
\begin{equation}
\limsup_{l\to\infty}\max_{s^l\in\bS^l} \int \frac{1}{M_l} \sum_{i=1}^{M_l} \tr\left(W_{s^l}(x_i^l)(\eins_{\hr^{\otimes l}} - D_i^l)\right)\ d\mu_l((x_i^l,D_i^l)_{i=1}^{M_l})=0.
\end{equation}
\end{definition}
\begin{definition}\label{def:avc_det-achievability}
A non-negative number $R$ is called achievable for transmission of messages over the AVcqC $\mathbb W=\{W_s\}_{s\in \bS}$ with deterministic codes using the average error criterion if
it is achievable with random codes by a sequence $(\mu_l)_{l\in\nn}$ which are deterministic codes.
\end{definition}
\begin{definition}
The capacity for message transmission using random codes and the average error criterion of an AVcqC $\mathbb W$ is given by
\begin{eqnarray}
\overline C_{\textup{random}}(\mathbb W):=\sup\left\{R:\begin{array}{l}R\textrm{\ is\ an\ achievable\ rate\ for\ transmission\ of\ messages\ over\ }\mathbb W\\
 \textrm{\ with\ random\ codes\ using\ the\ average\ error\ probability\ criterion}\end{array}\right\}.
\end{eqnarray}
\end{definition}
\begin{definition}
The capacity for message transmission using deterministic codes and the average error criterion of an AVcqC $\A$ is given by
\begin{eqnarray}
\overline C_{\textup{det}}(\mathbb W):=\sup\left\{R:\begin{array}{l}R\textrm{\ is\ an\ achievable\ rate\ for\ transmission\ of\ messages\ over\ }\mathbb W\textrm{\
with}\\
 \textrm{deterministic\ codes\ using\ the\ average\ error\ probability\ criterion}\end{array}\right\}.
\end{eqnarray}
\end{definition}
We will now come back to our basic object, the AVQC, and establish the notions for its various message transmission capacities:
\begin{definition}\label{def:message-trans-with-average-error}
Let $l\in\nn$. A random code for message transmission over $\fri$ is given by a probability measure $\gamma_l$ on the set $(CQ(M_l,\hr^{\otimes l})\times\mathcal M_{M_l},\Sigma_l)$, where $\Sigma_l$ again denotes a $\sigma-$algebra containing all singleton sets. Such $\sigma$-algebras exist, by arguments similar to those given in \cite{abbn} and\cite{bbjn}. A deterministic code is then given by a random code $\gamma_l$, where $\gamma_l$ is a point (also called Dirac) measure.
\end{definition}
\begin{definition}
A nonnegative number $R$ is called achievable with random codes under the average error criterion if there exists a sequence $(\gamma_l)_{l\in\nn}$ of random codes satisfying both
\begin{align}
1)&\qquad\liminf_{l\to\infty}\min_{s^l\in\bS^l}\int\frac{1}{M_l}\sum_{i=1}^{M_l}\tr\{D_i\cn_{s^l}(\cP(i))\}d\gamma_l(\cP,\mathbf D)=1\\
2)&\qquad\limsup_{l\to\infty}\frac{1}{l}\log M_l\geq R.
\end{align}
It is called achievable with random codes under the maximal error criterion if instead of $1)$ even
\begin{align}
1')&\qquad\liminf_{l\to\infty}\min_{s^l\in\bS^l}\min_{i\in[M_l]}\int\tr\{D_i\cn_{s^l}(\cP(i))\}d\gamma_l(\cP,\mathbf D)=1
\end{align}
holds.\\
If the sequence $(\gamma_l)_{l\in\nn}$ can be chosen to consist of point measures only, then $R$ is called achievable with deterministic codes under the average error criterion if $1)$ and $2)$ hold and it is called achievable with deterministic codes under the maximal error criterion if $1')$ and $2)$ hold.
\end{definition}
\begin{definition}
The corresponding capacities of $\fri$ are defined as
\begin{align}
\overline C_{\mathrm{det}}(\fri)&:=\sup\left\{R:
\begin{array}{l}
R\ \mathrm{is\ achievable\ with\ deterministic\ codes}\\
\mathrm{under\ the\ average\ error\ criterion}
\end{array}\right\},\\
\overline C_{\mathrm{random}}(\fri)&:=\sup\left\{R:
\begin{array}{l}
R\ \mathrm{is\ achievable\ with\ random\ codes}\\
\mathrm{under\ the\ average\ error\ criterion}
\end{array}\right\},\\
C_{\mathrm{det}}(\fri)&:=\sup\left\{R:
\begin{array}{l}
R\ \mathrm{is\ achievable\ with\ deterministic\ codes}\\
\mathrm{under\ the\ maximal\ error\ criterion}
\end{array}\right\},\\
C_{\mathrm{random}}(\fri)&:=\sup\left\{R:
\begin{array}{l}
R\ \mathrm{is\ achievable\ with\ random\ codes}\\
\mathrm{under\ the\ maximal\ error\ criterion}
\end{array}\right\}.
\end{align}
\end{definition}
From \cite{ahlswede-elimination}, \cite{abbn} and \cite{bbjn} it is clear that common randomness is a useful resource. Readers with a deeper interest in the topic will find it fruitful to take a look at Theorem 1 and Theorem 2 a) in \cite{ahlswede-elimination}) or Theorem 32 and Lemma 37 in \cite{abbn} or Lemma 9 and Lemma 10 in \cite{bbjn} for applications to AVQCs and AVcqCs.\\
The proofs given there rely on the possibility to establish approximately perfect correlations between sender and receiver. This kind of correlation is called 'common randomness' here (although, to be fair, one should say that every $p\in\mathfrak P(\bX\times\bY)$ satisfying $p\neq p_\bX\times p_\bY$ deserves that name). As explained in the introduction, terms like 'common randomness' or 'common information' have a variety of different definitions attached to it. We will state the definitions that are relevant for our work now.
\begin{definition}\label{def:source of common randomness}
A source of common randomness $CR\geq0$ is given by a sequence $(\gamma_l)_{l\in\nn}$ of probability distributions, where $\gamma_l\in\mathfrak P(\Gamma_l\times\Gamma_l)$ for every $l\in\nn$ and, asymptotically, we have
\begin{align}
1)&\qquad\liminf_{l\to\infty}\frac{1}{l}\log|\Gamma_l|=CR\\
2)&\qquad\limsup_{l\to\infty}\|\gamma_l-\bar\delta_l\|_1=0,
\end{align}
where $\bar\delta_l\in\mathfrak P(\Gamma_l\times\Gamma_l)$ denotes the normalized delta function, $\bar\delta_l(i,j)=1/|\Gamma_l|$ if $i=j$ and $\bar\delta_l(i,j)=0$ else.
\end{definition}
\begin{definition}
Let $\bX$ and $\bY$ be finite alphabets. A probability distribution $p\in\mathfrak P(\bX\times\bY)$ is said to have common randomness $CR\geq0$ if there exists a sequence $(f_l,g_l)_{l\in\nn}$ of functions such that for every $l\in\nn$ $f_l:\bX^l\mapsto\Gamma_l$, $g_l:\bY^l\mapsto\Gamma_l$ with $\Gamma_l$ being a finite set and asymptotically we have that
\begin{align}
1)&\qquad\liminf_{l\to\infty}\frac{1}{l}\log|\Gamma_l|=CR\\
2)&\qquad\limsup_{l\to\infty}\|(f_l\times g_l)\circ p^{\otimes l}-\bar\delta_l\|_1=0.
\end{align}
The supremum over all $CR$ such that $p$ has $CR$ is called the common randomness of $p$ and is written $CR(p)$.
\end{definition}
As it turns out in Theorem \ref{theorem:no-common-randomness-extraction}, the set of probability distributions on some $\bX\times\bY$ with a positive common randomness is exceptional in a lot of ways: Its complement is open and dense in $\mathfrak P(\bX\times\bY)$, its measure (with respect to the Lebesgue measure) is zero, and operationally it is highly unstable with respect to small perturbations.\\
Let us now define correlated codes for message transmission over arbitrarily varying quantum channels.
\begin{definition}
Given $l\in\nn$, a random variable $(X,Y)$ distributed according to $p\in\mathfrak P(\bX\times\bY)$ for some finite sets $\bX,\bY$ and $r\in\nn$, an $((X,Y),r)-$correlated code $\mathfrak C_l$ for message transmission over $\fri$ is given by a set
\begin{align}
\mathfrak C_l=\{(M_l,\mathcal{P}_{x^{n(l)}},\mathcal D_{y^{n(l)}})\}_{x^{n(l)}\in\bX^{n(l)},y^{n(l)}\in\bY^{n(l)}},
\end{align}
where:
\begin{enumerate}
\item $M_l\in\nn$,
\item $\mathcal{P}_{x^{n(l)}}\in CQ([M_l],\hr^{\otimes l})$ for all $x^{n(l)}\in\bX^{n(l)}$,
\item $\mathcal D_{y^{n(l)}}=\{D_{y^{n(l)},1},\ldots,D_{y^{n(l)},M_l}\}$ is a POVM on $\kr^{\otimes l}$ for all $y^{n(l)}\in\bY^{n(l)}$,
\item and $n(l):=\lfloor l/r\rfloor$.
\end{enumerate}
\end{definition}
\begin{definition}
Let $(X,Y)$ be a random variable distributed according to $p\in\mathfrak P(\bX\times\bY)$ and $r\in\nn$. A nonnegative number $R$ is said to be an achievable m-a-((X,Y),r) (message transmission under the average error criterion using $((X,Y),r)-$correlated codes) rate for $\fri$ if there exists a sequence $(\mathfrak C_l)_{l\in\nn}$ of $((X,Y)r)-$correlated codes such that
\begin{align*}
1)\qquad & \liminf_{l\to\infty}\frac{1}{l}\log M_l\geq R\\
2)\qquad & \lim_{l\to\infty}\inf_{s^l\in\bS^l}\sum_{x^{n(l)}\in\bX^{n(l)}}\sum_{y^{n(l)}\in\bY^{n(l)}}p^{\otimes {n(l)}}(x^{n(l)},y^{n(l)})\frac{1}{M_l}\sum_{i=1}^{M_l}\tr\{D_{y^{n(l)},i}\cn_{s^l}(\cP_{x^{n(l)}}(i))\}=1
\end{align*}
\end{definition}
\begin{definition}
Given the random variable $(X,Y)$ distributed according to $p\in\mathfrak P(\bX\times\bY)$ and an $r\in\nn$, the m-a-((X,Y),r) capacity $\overline C(\fri,r,(X,Y))$ of the AVQC $\fri$ is given by
\begin{align*}
\overline C(\fri,r,(X,Y)):=\sup\{R:R\ \mathrm{is\ achievable\ }m-a-((X,Y),r)\ \mathrm{rate\ for}\ \fri\}.
\end{align*}
\end{definition}
\begin{remark}
Obviously, the following inequalities hold true.
\begin{align}
C_{\mathrm{det}}(\fri)\leq\overline C_{\mathrm{det}}(\fri)\leq\overline C_{\mathrm{rand}}(\fri),\\
C_{\mathrm{det}}(\fri)\leq C_{\mathrm{rand}}(\fri)\leq\overline C_{\mathrm{rand}}(\fri),\\
\overline C(\fri,r,(X,Y))\leq\overline C_{\mathrm{rand}}(\fri).
\end{align}
\end{remark}
During the proof of our main result, Theorem \ref{thm-1}, the following objects will play a central role.
\begin{definition}\label{definition-associated-avcqc}
Given $n\in\nn$ and $\fri\subset\mathcal C(\hr,\kr)$, let $\hr_{Y^n}$ be a Hilbert space of dimension $|\bY|^n$ and $\{\hat\rho_{y^n}\}_{y^n\in\bY^n}\subset\cs(\hr_{Y^n})$ be a set of pairwise orthogonal and pure states.
For a set $\mathfrak S_K=\{\rho_1,\ldots,\rho_K\}\subset\cs(\hr^{\otimes n})$, the associated AVcqC
$\mathbb W_{\mathfrak S_K}=\{W_{s^n}\}_{s^n\in\bS^n}\subset CQ(\mathbb F(\bX^n,[K]),\hr_{\bY^{n}}\otimes\kr^{\otimes n})$ is defined by
\begin{align}
W_{s^n}(f):=\sum_{x^n,y^n}p^{\otimes n}(x^n,y^n)\cdot\hat\rho_{y^n}\otimes\cn_{s^n}(\rho_{f(x^n)})\qquad (s^n\in\bS^n).
\end{align}
Here, $\mathbb F(\bX^n,[K])$ denotes the functions on $\bX^n$ taking values in $[K]$.
\end{definition}
\end{section}
\begin{section}{\label{sec:Main Results}Main Results}
This section enlists our main results, in order of appearance. First, in Section \ref{sec:common-randomness}, we prove a result concerning common randomness:
\begin{theorem}\label{theorem:no-common-randomness-extraction}
Let a bipartite classical source $(X,Y)$ with values in $\bX\times\bY$ which is distributed according to $\mathbb P(X=x,Y=y)=p(x,y)\ \forall (x,y)\in\bX\times\bY$ be given, where $p\in\mathrm{ri}\mathfrak P(\bX\times\bY)$. There is no sequence $(f_l,g_l)_{l\in\nn}$ of functions $f_l:\bX^l\to\Gamma_l$, $g_l:\bY^l\to\Gamma_l$ ($l\in\nn$) satisfying
\begin{align*}
(1)\qquad&\lim_{l\to\infty}|\Gamma_l|=\infty,\\
(2)\qquad&\lim_{l\to\infty}p^{\otimes l}(\{(x^l,y^l):f_l(x^l)=g_l(y^l)\})=1,\\
(3)\qquad&\lim_{l\to\infty}p_\bX^{\otimes l}(\{x^l:f_l(x^l)=k_l\})=\lim_{l\to\infty}p_\bY^{\otimes l}(\{y^l:g_l(x^l)=k_l\})=0\qquad\forall\ (k_l)_{l\in\nn}\subset\bigtimes_{l=1}^\infty\Gamma_l.
\end{align*}
Further, the set of probability distributions $p\in\mathfrak P(\bX\times\bY)$ satisfying $CR(p)>0$ is closed.
\end{theorem}
\begin{remark}
Theorem \ref{theorem:no-common-randomness-extraction} tells us not only that, for given $\bX,\bY$, $CR(p)_{\upharpoonright\mathrm{ri}\mathfrak P(\bX\times\bY)}=0$, but also that not even a \emph{polynomially} small amount of common randomness can be extracted!\\
It further states that every point at which $CR$ (as a function on $\mathfrak P(\bX\times\bY)$) is positive is also a point at which it is discontinuous, underlining the fragile character of this resource.
\end{remark}
The importance of this statement stems from the strategy of proof that is used in \cite{ahlswede-elimination}, \cite{abbn}, \cite{bbjn} in order to establish the Ahlswede Dichotomy (the original formulation can be found in \cite{ahlswede-elimination}, Theorem 1) in its various forms, namely for message transmission over classical channels using the average error criterion, entanglement transmission over quantum channels and message transmission over cq channels using the average error criterion.\\
In order to convert random codes into deterministic codes, in every of the three scenarios it turns out that establishing a small (only polynomial in the number of channel uses) amount of common randomness first is useful. We emphasize that, although an explicit example of the advantage that random codes offer over their deterministic counterparts has been given in \cite{ahlswede-blinovsky}, it is conjectured in \cite{abbn} that $\A_{\mathrm{det}}(\fri)=\A_{\mathrm{random}}(\fri)$ holds for every AVQC $\fri$. To complete the list of models and capacities that one could investigate in this matter we point out that no such example has been given for classical message transmission over AVQCs, although one should expect that the case $\overline C_{\mathrm{random}}(\fri)>\overline C_{\mathrm{det}}(\fri)$ occurs.\\
In any case, by Theorem \ref{theorem:no-common-randomness-extraction} it is clear that common randomness is a costly resource. This underlines the importance of the next two theorems, our main results, which are proven in Sections \ref{sec:correlation-and-messages} and \ref{sec:correlation-and-entanglement}.
\begin{theorem}\label{thm-2}
Let $\fri=\{\cn_s\}_{s\in\bS}\subset\mathcal C(\hr,\kr)$ denote a finite AVQC. If $(X,Y)$ is a random variable taking values in the finite alphabet $\bX\times\bY$ such that $I(X,Y)>0$, then for every $r\in\nn$ we have $\overline C(\fri,r,(X,Y))=\overline C_{\mathrm{random}}(\fri)$.
\end{theorem}
It will be clear from the proof of Theorem \ref{thm-2} that the corresponding results hold w.l.o.g. for AVcqCs as well, hence with the notation adapted in the obvious way we get
\begin{corollary}\label{corollary-2}
Let $\mathbb W=\{W_s\}_{s\in\bS}\subset CQ(\mathbf Z,\kr)$ denote a finite AVcqC. If $(X,Y)$ is a random variable taking values in the finite alphabet $\bX\times\bY$ such that $I(X,Y)>0$, then for every $r\in\nn$ we have $\overline C(\mathbb W,r,(X,Y))=\overline C_{\mathrm{random}}(\mathbb W)$.
\end{corollary}
\begin{theorem}\label{thm:entanglement-transmission-over-avqc}
Let $\fri=\{\cn_s\}_{s\in\bS}\subset\mathcal C(\hr,\kr)$ denote a finite AVQC. If $(X,Y)$ is a random variable taking values in the finite alphabet $\bX\times\bY$ such that $I(X,Y)>0$, then $\A(\fri,r,(X,Y))=\A_{\mathrm{random}}(\fri)$ for all $r\in\nn$.
\end{theorem}
Another important result is the equivalence of maximal - and average error criterion for AVQCs. This result shoud be compared to Theorem 2 a) and Theorem 3 a) in \cite{ahlswede-elimination}. The capacity for message transmission over an arbitrarily varying channel does in general depend on which of the two criteria one uses. However, in the case of randomized encoding, the two capacities coincide.\\
The codes that are used in our definition of the two capacities allow for a randomized encoding, since the signal states that get fed into the channel at senders side are allowed to be mixed. Taking this into account, the following result is not too surprising:
\begin{theorem}\label{thm:equivalence-of-average-and-maximal-error}
Let $\fri$ be a finite AVQC. Then $\overline C_{\mathrm{det}}(\fri)=C_{\mathrm{det}}(\fri)$.
\end{theorem}
In applications it might be necessary to check whether a given AVQC has a positive capacity for transmission of classical messages or not. This can be done by checking whether it is $l-$symmetrizable for every $l\in\nn$ or not (see Definition 39, Theorem 40 in \cite{abbn}). We provide a slightly less complicated formula in Theorem \ref{thm:pure-state-symmetrizability}. Even though it has a simplified structure, Theorem \ref{thm:pure-state-symmetrizability} does not yet provide a substantial benefit over the original formulation, at least when it comes to calculations.\\
However, we found yet another equivalent reformulation of the notion of $l-$symmetrizability:
\begin{theorem}\label{thm:geometric-variant-of-symmetrizability}
Let $l\in\nn$ and $A_1,\ldots,A_K$ be a set of operators such that $\cs(\hr^{\otimes l})\subset\conv(\{A_i\}_{i=1}^K)$. A finite AVQC $\fri=\{\cn_s\}_{s\in\bS}$ is $l-$symmetrizable if and only if there is a set $\{p_i\}_{i=1}^K\subset\mathfrak P(\bS^l)$ of probability distributions such that
\begin{align}
\sum_{s^l\in\bS^l}q_j(s^l)\cn_{s^l}(A_i)=\sum_{s^l\in\bS^l}q_i(s^l)\cn_{s^l}(A_j)\qquad\forall i,j\in[K].
\end{align}
\end{theorem}
Theorem \ref{thm:geometric-variant-of-symmetrizability} not only provides a more handy criterion, it also provides a more geometric view of the symmetrizability condition.\\
We did not pursue the question further, but $K=d^{2l}$ should be a sufficient number of operators $A_i$ in above theorem, corresponding to an embedding of the set of states into a simplex.
\end{section}
\begin{section}{\label{sec:common-randomness}High Price for Common Randomness}
In a first step, we give minimal requirements for the kind of common randomness we can put to use in order to achieve a message transmission capacity $\overline C_{\mathrm{random}}(\fri)$ or an entanglement transmission capacity of $\A_{\mathrm{random}}(\fri)$. Please notice that, as explained in Section \ref{sec:Main Results}, it is still an open question whether common randomness offers any advantage \emph{at all} for entanglement transmission over AVQCs. In contrast to this assumption, it has been proven in \cite{ahlswede-blinovsky} that common randomness \emph{is} advantageous for AVcqCs!\\
Note that we can safely assume $\overline C_{\mathrm{det}}(\fri)=0$, since otherwise we already achieve above numbers (in case of entanglement transmission, this is stated and proven as Theorem 5 in \cite{abbn}, for message transmission it is obvious for the expert but yet unproven), even without the use of common randomness.
\begin{lemma}\label{lemma:necessary-condition-for-random-code}
Let $\fri=\{\cn_s\}_{s\in\bS}$ be a symmetrizable AVQC. A sequence $(\gamma_l)_{l\in\nn}$ of finitely supported random codes can lead to a positive message transmission rate over $\fri$ only if it satisfies
\begin{align}
1)&\qquad \liminf_{l\to\infty}\gamma_l(k_l,k_l')=0,\\
2)&\qquad \liminf_{l\to\infty}\gamma_{S,l}(k_l)=0,\\
3)&\qquad \liminf_{l\to\infty}\gamma_{R,l}(k_l')=0
\end{align}
for all possible sequences $(k_l,k_l')_{l\in\nn}$ satisfying $k_l,k_l'\in\Gamma_l$ for every $l\in\nn$ and with $\gamma_{S,l}$ and $\gamma_{R,l}$ being the marginal distributions of $\gamma_l$ at sender and receivers side ($l\in\nn$).
\end{lemma}
\begin{remark}
It follows directly from this lemma that no finite amount of common randomness can be sufficient in order to achieve positive message transmission capacity over a symmetrizable AVQC. It is also clear that both common randomness and any sequence $(p^{\otimes l})_{l\in\nn}$ ($p\in\mathfrak P(\bX\times\bY)$) satisfy these conditions, but that common randomness has one additional property (see the forthcoming equations (\ref{eqn:201})).
\end{remark}
\begin{proof}
We shall deal with $1)$ first. Suppose there is a sequence $(k_l,k_l')_{l\in\nn}$ satisfying $k_l,k_l'\in\Gamma_l$ for every $l\in\nn$ such that $\liminf_{l\to\infty}\gamma_l(k_l,k_l')=c>0$ and $M_l\geq2$ for large enough $l\in\nn$. Without loss of generality, $k_l=k_l'=1$ for all $l\in\nn$. Given that $\fri$ is symmetrizable (which is equivalent to $\overline C_{\mathrm{det}}(\fri)=0$), the proof of Theorem 40 in \cite{abbn} shows that there exists a sequence $(\mathbf s^l)_{l\in\nn}$ (each $\mathbf s^l$, $l\in\nn$, being an element of $\bS^l$) for which
\begin{align}
\frac{1}{M_l}\sum_{i=1}^{M_l}\tr\{D_{1,i}\cn_{\mathbf s^l}(\rho_{1,i})\}\leq3/4
\end{align}
holds, hence
\begin{align}
\liminf_{l\to\infty}\sum_{(k_l,k_l')\in\Gamma_l\times\Gamma_l}\frac{\gamma_l(k_l,k_l')}{M_l}\sum_{i=1}^{M_l}\tr\{D_{j,i}\cn_{\mathbf s^l}(\rho_{j,i})\}&\leq\liminf_{l\to\infty}[\frac{\gamma_l(1,1)}{M_l}\sum_{i=1}^{M_l}\tr\{D_{1,i}\cn_{\mathbf s^l}(\rho_{1,i})\}+(1-\gamma_l(1,1))]\\
&\leq\liminf_{l\to\infty}[\gamma_l(1,1)3/4+(1-\gamma_l(1,1))]\\
&=1-\liminf_{l\to\infty}\gamma_l(1,1)/4\\
&\leq1-c/4\\
&<1.
\end{align}
Aiming at $2)$ and using the distributions $\gamma_l(\cdot|k_l)\in\mathfrak P(\Gamma_l)$ $(k_l\in\Gamma_l)$ defined by
\begin{align}
\gamma_l(k_l'|k_l)&:=\left\{
\begin{array}{ll}
\gamma_l(k_l,k_l')/\gamma_{S,l}(k_l), & \mathrm{if}\ \gamma_l(k_l,k_l')>0\\
0, & \mathrm{else}
\end{array}\right.& (k_l,k_l'\in\Gamma_l)
\end{align}
we define new decoding POVMs by
\begin{align}
\overline D_{k_l,i}:=\sum_{k_l'}\gamma_l(k_l'|k_l)D_{k_l',i}.
\end{align}
Then
\begin{align}
\liminf_{l\to\infty}\min_{s^l\in\bS^l}\sum_{k_l\in\Gamma_l}\gamma_{S,l}(k_l)\frac{1}{M_l}\sum_{i=1}^{M_l}\tr\{\overline D_{k_l,i}\cn_{s^l}(\rho_{k_l,i})\}=1
\end{align}
Application of the same trick as in case $1)$ leads to statement $2)$. The proof of $3)$ follows a reasoning along the lines of the proof of $2)$ and will therefore be omitted.
\end{proof}
A look at Lemma 37 in \cite{abbn} or Lemma 10 in \cite{bbjn} will convince the reader that, in order to transmit either entanglement or classical messages over a symmetrizable AVQC it suffices to use only a small (polynomial in the number of channel uses) amount of common randomness in order to guarantee transmission at the respective random capacity.\\
It is also immediately clear that in both cases one can get this result by using only a source of common randomness according to Definition \ref{def:source of common randomness} with $|\Gamma_l|=l^2$ for all $l\in\nn$. But these trivially satisfy the following.
\begin{align}\label{eqn:201}
\forall\ (k_l,k_l')_{l\in\nn}\subset\bigtimes_{l=1}^\infty\Gamma_l\times\Gamma_l:\qquad
\begin{array}{ll}
\gamma_{R,l}(k_l')\underset{l\to\infty}\longrightarrow0,&\qquad\gamma_{S,l}(k_l)\underset{l\to\infty}\longrightarrow0\\
\gamma_l(k_l,k_l')\underset{l\to\infty}{\longrightarrow}0,&\qquad\sum_{k_l\in\Gamma_l}\gamma_l(k_l,k_l)\underset{l\to\infty}{\longrightarrow}1.
\end{array}
\end{align}
Next, we use the so far obtained results in order to prove Theorem \ref{theorem:no-common-randomness-extraction}.
\begin{proof}[Proof of Theorem \ref{theorem:no-common-randomness-extraction}]
Assume there is such sequence of functions. For values $k\in\Gamma_l$ ($l\in\nn$ arbitrary) we use the following abbreviations.
\begin{align}
a_l(k)&:=p_\bX^{\otimes l}(\{x^l:f_l(x^l)=k\})\\
b_l(k)&:=p_\bY^{\otimes l}(\{y^l:g_l(y^l)=k\})\\
c_l(k)&:=p^{\otimes l}(\{(x^l,y^l):f_l(x^l)=g_l(y^l)=k\})
\end{align}
Let $\eps>0$ and $l\in\nn$ satisfy
\begin{align}
1-\eps\leq\sum_{k=1}^{|\Gamma_l|}c_l(k)\leq1,\qquad a_l(k)\leq\eps,\qquad b_l(k)\leq\eps\qquad \forall k\in\Gamma_l.
\end{align}
Note that, upon choosing $l$ to be large enough, we can make $\eps$ arbitrarily small (compare equations (\ref{eqn:201})). Consider the monotone increasing sequences $(A_m)_{m=1}^{|\Gamma_l|}$ and $(B_m)_{m=1}^{|\Gamma_l|}$ defined by
\begin{align}
A_m:=\sum_{k=1}^ma_l(k),\qquad B_m:=\sum_{k=1}^mb_l(k).
\end{align}
Let $0<\sigma<1/2$ and $\hat m$ be the smallest number such that
\begin{align}
\min\{A_{\hat m},B_{\hat m}\}\geq\sigma.
\end{align}
Assume (w.l.o.g.) that $A_{\hat m}\leq B_{\hat m}$ holds. We have
\begin{align}
B_{\hat m}&=1-\sum_{m=\hat m+1}^{|\Gamma_l|}b_l(m)\\
&\leq1-\sum_{m=\hat m+1}^{|\Gamma_l|}c_l(m)\\
&\leq\eps+\sum_{m=1}^{\hat m}c_l(m)\\
&=\eps+\sum_{m=1}^{\hat m-1}c_l(m)+c_l(\hat m)\\
&\leq\eps+\sum_{m=1}^{\hat m-1}a_l(m)+a_l(\hat m)\\
&\leq2\eps+\sigma.
\end{align}
It follows that
\begin{align}
\sigma\leq\min\{A_{\hat m},B_{\hat m}\}\leq\max\{A_{\hat m},B_{\hat m}\}\leq\sigma+2\eps.
\end{align}
Let $\Theta_l:\Gamma_l\rightarrow\{0,1\}$ be defined by
\begin{align}
\Theta_l(k):=\left\{
\begin{array}{ll}
1, & \mathrm{if}\ 1\leq k\leq\hat m\\
0, & \mathrm{else}
\end{array}\right.
\end{align}
Then
\begin{align}
p^{\otimes l}(\Theta_l\circ f_l=\Theta_l\circ g_l)&=p^{\otimes l}(\{(x^l,y^l):f_l(x^l),g_l(y^l)\in[\hat m]\})+p^{\otimes l}(\{(x^l,y^l):f_l(x^l),g_l(y^l)\notin[\hat m]\})\\
&\geq\sum_{k=1}^{|\Gamma_l|}c_l(k)\geq1-\eps
\end{align}
and
\begin{align}
p^{\otimes l}(\Theta_l\circ f_l\neq\Theta_l\circ g_l)=1-p^{\otimes l}(\Theta_l\circ f_l=\Theta_l\circ g_l)\leq\eps
\end{align}
Also, for $\eps$ small enough ($\eps\leq(1-2\sigma)/2$ is sufficient) we get
\begin{align}
p^{\otimes l}_\bX(\Theta_l\circ f_l=1)&=A_{\hat m}\in[\sigma,\sigma+2\eps]\subset[\sigma,1-\sigma]\\
p^{\otimes l}_\bY(\Theta_l\circ g_l=1)&=B_{\hat m}\in[\sigma,\sigma+2\eps]\subset[\sigma,1-\sigma].
\end{align}
But this is impossible by (\cite{witsenhausen}, 1. Problem statement), since $p\in\mathrm{ri}\mathfrak P(\bX\times\bY)$ implies that there is $\eps>0$ such that $p(A)>\eps$ for every $A\subset\bX\times\bY$.\\
Under such condition, no partition of $\bX$ and $\bY$ into $k\geq2$ disjoint sets $\bX_1,\ldots,\bX_k$ and $\bY_1,\ldots,\bY_k$ satisfying $p_\bX(\bX_i)>0,\ p_\bY(\bY_i)>0$ for all $i=1,\ldots,k$ can satisfy $p(\cup_{i=1}^k \bX_i\times\bY_i)=1$, since it holds $\bX\times\bY\cap(\cup_{i=1}^k\bX_i\times\bY_i)^\complement\neq\emptyset$.\\
It follows from \cite{witsenhausen} that the sequence $(f_l,g_l)_{l\in\nn}$ satisfying the three conditions in Lemma \ref{theorem:no-common-randomness-extraction} cannot exist.\\
At last we will now show that the set of probability distributions $P\in\mathfrak P(\bX\times\bY)$ on which $CR(p)>0$ holds is closed. This is seen as follows. For a given sequence $(p_n)_{n\in\nn}$ with entries taken from $\mathfrak P(\bX\times\bY)$ and such that for every $n\in\nn$ we have $CR(p_n)>0$ we know from \cite{witsenhausen} that for each such $p_n$ there is a decomposition of $\bX,\bY$ into disjoint sets $\bX_{1,n},\ldots,\bX_{k_n,n}$, $\bY_{1,n},\ldots,\bY_{k_n,n}$ with $k_n\geq2$ such that $p_n(\cup_{i=1}^{k_n}\bX_{i,n}\times\bY_{i,n})=1$ holds.\\
Now let $p\in\mathfrak P(\bX\times\bY)$ satisfy $\lim_{n\to\infty}\|p-p_n\|_1=0$. This is equivalent to the existence of a sequence $\eps_n\searrow0$ such that $|p_n(x,y)-p(x,y)|\leq\eps_n$ holds for all $n\in\nn$ and $(x,y)\in\bX\times\bY$.
Thus if $p(x,y)>0$ for some $(x,y)\in\bX\times\bY$ then necessarily $p_n(x,y)>0$ for all large enough $n\in\nn$, hence there is $N\in\nn$ such that
\begin{align}
\supp(p)\subset\supp(p_n)\qquad\forall n\geq N,
\end{align}
hence $p$ inherits all the decompostions $\cup_{i=1}^{k_n}\bX_{i,n}\times\bY_{i,n}$ for which $n\geq N$. By assumption, all of them were nontrivial ($k_n\geq2$), so $CR(p)>0$ as desired.
\end{proof}
\end{section}
\begin{section}{\label{sec:correlation-and-messages}Impact of Correlation on Messsage Transmission over an AVQC}
We will first argue that in the proofs of Theorems \ref{thm-1}, \ref{thm:entanglement-transmission-over-avqc} and \ref{thm-2} we can safely assume that $\bX=\bY=\{0,1\}$ holds.\\
To be more specific, given an arbitrary $(X',Y')$ with $I(X',Y')>0$, sender and receiver can always process the outcomes such that they are effectively left with a binary bipartite source $(X,Y)$ satisfying $I(X,Y)>0$. Thus, for sake of simplicity, in proofs we will always assume that a source is binary, although our statements hold for arbitrary bipartite sources.\\
For readers convenience, a rigorous justification of above assumption that has also been made in \cite{ahlswede-cai} is given now.\\
Let $\bX',\bY'$ be finite alphabets and a random variable $(X',Y')$ with values in $\bX'\times\bY'$ be given that is distributed according to $p'\in\mathfrak P(\bX'\times\bY')$. Let $I(X',Y')>0$. Then there are functions $f:\bX\rightarrow\{0,1\}$ and $g:\bY\rightarrow\{0,1\}$ such that the random variable $(X,Y):=(f(X'),g(Y'))$ with values in $\{0,1\}\times\{0,1\}$ satisfies $I(X,Y)>0$. This is seen as follows.\\
Assume there are no two such functions.\\
Then, considering pre-images of elements $a,b\in\{0,1\}$ under all possible $f$ and $g$ and with $p'_\bX$ and $p'_\bY$ denoting the usual marginal distributions of $p'$ it must hold that
\begin{align}
p'(A,B)=p'_\bX(A)\cdot p'_\bY(B)\qquad\forall A\subset\bX,\ B\subset\bY.
\end{align}
Especially, this would hold for all one-element sets, implying $p'(a,b)=p'_\bX(a)\cdot p'_\bY(b)$ for all $(a,b)\in\bX\times\bY$. But this latter equality contradicts the assumption $I(X',Y')>0$, so above mentioned 'pre-processed' random variable $(X,Y)$ with $I(X,Y)>0$ exists.
\\\\
Let us now start our investigations. The following Theorem \ref{thm-1} will enable us to prove Theorem \ref{thm-2} by using the correlation present in $p$ to first establish some common randomness between sender and receiver and then operate a randomness-assisted code close to the random capacity afterwards, a strategy that has first been successfully used in \cite{ahlswede-elimination} and, afterwards, in \cite{abbn} and \cite{bbjn}.
\begin{theorem}\label{thm-1}
If $I(X,Y)>0$ for some bipartite source $p\in\mathfrak P(\bX\times\bY)$ and $\overline C(\fri,r,(X,Y))=0$ for any $r\in\nn$, then for every $l\in\nn$ and set $\mathfrak S_K=\{\rho_1,\ldots,\rho_K\}\subset\cs(\hr^{\otimes rl})$ the associated AVcqC
$\mathbb W^l_{\mathfrak S_K}$ satisfies $\overline C(\mathbb W^l_{\mathfrak S_K})=0$. Most notably, this implies that $\overline C_{\mathrm{random}}(\fri)=0$.
\end{theorem}
\begin{proof} Let us, for clarity of the argument, first assume that $r=1$.
For every sequence $\mathfrak C_l$ of correlated codes with
\begin{align}\liminf_{l\to\infty}\frac{1}{l}\log M_l>0,\ \
\lim_{l\to\infty}\inf_{s^l\in\bS^l}\sum_{x^l\in\bX^l}\sum_{y^l\in\bY^l}p^{\otimes l}(x^l,y^l)\frac{1}{M_l}\sum_{i=1}^{M_l}\tr\{D_{y^l,i}\cn_{s^l}(\cP_{x^l}(i))\}<1\ \ \mathrm{holds}.
\end{align}
This especially holds for any given set $\mathfrak S_K=\{\rho_1,\ldots,\rho_K\}\subset\cs(\hr)$ and corresponding sequences of correlated codes with product encodings $\cP_{x^l}(i):=\otimes_{j=1}^l\rho_{f^{(i)}_j(x_j)}$ ($i=1,\ldots,M_l$, $j=1,\ldots,l$,
$f^{(i)}_j:\bX\rightarrow[K]$).\\
Now consider the AVcqC $\mathbb W_{\mathfrak S_K}$ as in Definition \ref{definition-associated-avcqc}. Any POVM $\mathbf{D}\in\M_{M_l}(\hr_{Y^l}\otimes\kr^{\otimes l})$ gives rise to a set $\{\mathbf D_{y^l}\}_{y^l\in\bY^l}$ of POVMs in $\M_{M_l}(\kr^{\otimes l})$ through
\begin{align}
D_{y^l,i}:=\tr_{\hr_{Y^l}}\otimes Id_{\kr^{\otimes l}}\{(\hat\rho_{y^l}\otimes\eins_{\kr^{\otimes l}})D_i\},
\end{align}
while any encoding at blocklength $l$ for $\mathbb W_{\mathfrak S_K}$ has the form
\begin{align}i\mapsto f^{(i)}_1\times\ldots\times f^{(i)}_l\qquad \forall l\in\nn.\end{align}
for some choice of functions $\{f^{(i)}_j\}_{i\in[M_l],j\in[l]}\subset \mathbb F(\bX^l,[K])$. Then for all $l\in\nn$,
\begin{align}
\sum_{i=1}^{M_l}\tr\{D_i W_{s^l}(\bigtimes_{j=1}^lf^{(i)}_j)\}&=\sum_{i=1}^{M_l}\tr\{D_{i}\bigotimes_{j=1}^lW_{s_j}(f^{(i)}_j)\}\\
&=\sum_{i=1}^{M_l}\sum_{y^l\in\bY^l}\sum_{x^l\in\bX^l}p^{\otimes l}(x^l,y^l)\tr\{D_{y^l,i}\cn_{s^l}(\otimes_{j=1}^l\rho_{f^{(i)}_j(x_j)})\}\\
&=\sum_{i=1}^{M_l}\sum_{y^l\in\bY^l}\sum_{x^l\in\bX^l}p^{\otimes l}(x^l,y^l)\tr\{D_{y^l,i}\cn_{s^l}(\cP_{x^l}(i))\}.
\end{align}
Thus, no code for the AVcqC $\mathbb W_{\mathfrak S_K}$ can have asymptotically vanishing average error \emph{and} positive rate at the same time, hence
$\overline C(\mathbb W_{\mathfrak S_K})=0$ for every such AVcqC.\\
This proves the first part of the statement.\\
For the second, consider Theorem 1 from \cite{ahlswede-blinovsky}: Since every of the AVcqC's $\mathbb W_{\mathfrak S_K}$ has zero capacity for transmission of messages using
the average error criterion, it is symmetrizable in the sense of \cite{ahlswede-blinovsky}, and hence it follows that for every such $\mathbb W_{\mathfrak S_K}$ there exists
$\{\tau(\cdot|f)\}_{f\in\mathbb F(\bX,[K])}\subset \mathfrak P(\bS)$ such that
\begin{align}
\sum_{s\in\bS}\tau(s|f)W_s(f')=\sum_{s\in\bS}\tau(s|f')W_s(f)\qquad \forall f,f'\in\mathbb F(\bX,[K]).
\end{align}
Consider the functions $f_1,\ldots,f_K,f^*:\bX\rightarrow[K]$ defined by $f_i(0)=i,\ f_i(1)=i\oplus1$ ($\oplus$ denotes addition $\mathrm{mod}\ K$) and $f^*(0)=f^*(1)=1$.\\
Inserting these, we get for every $i=1,\ldots,K$ the following equalities:
\begin{align}
\sum_{s\in\bS}\tau(s|f_i)W_s(f^*)&=\sum_{s\in\bS}\tau(s|f^*)W_s(f_i)\\
\sum_{x\in\bX}\sum_{y\in\bY}p(x,y)\cdot\hat\rho_{y}\otimes\sum_{s\in\bS}\tau(s|f_i)\cn_s(\rho_1)&=\sum_{x\in\bX}\sum_{y\in\bY}p(x,y)\cdot\hat\rho_{y}\otimes\sum_{s\in\bS}\tau(s|f^*)\cn_s(\rho_{x\oplus i})
\end{align}
Define $p(\cdot|0),\ p(\cdot|1)\in\mathfrak P(\bX)$ by $p(\cdot|y):=p(\cdot,y)/(p(0,y)+p(1,y))$, $y=0,1$. Then $p(\cdot|0)=p(\cdot|1)$ if and only if $I(X,Y)=0$. Thus by assumption,
$p(\cdot|0)\neq p(\cdot|1)$. It follows
\begin{align}
\sum_{s\in\bS}\tau(s|f_i)\cn_s(\rho_1)&=\sum_{x\in\bX}p(x|y)\sum_{s\in\bS}\tau(s|f^*)\cn_s(\rho_{x\oplus i}),\qquad y=0,1,\ i\in[K],
\end{align}
hence setting
\begin{align}
\tilde\sigma_i:=\sum_{s\in\bS}\tau(s|f_i)\cn_s(\rho_1),\qquad \sigma_i:=\sum_{s\in\bS}\tau(s|f^*)\cn_s(\rho_{i})\qquad\forall i\in[K],
\end{align}
we know that for every $i\in\{1,\ldots,K\},$ the two states $\sigma_i,\ \sigma_{i\oplus1}$
satisfy both
\begin{align}
\tilde\sigma_i=p(0|0)\sigma_i+p(1|0)\sigma_{i\oplus1},\qquad \mathrm{and}\qquad \tilde\sigma_i=p(0|1)\sigma_i+p(1|1)\sigma_{i\oplus1}.
\end{align}
This can only be if $\sigma_i=\sigma_j$ f.a. $i=1,\ldots,K$ and hence the cq-channel $W^*\in CQ([K],\kr)$ defined by $W^*(i):=\sum_{s\in\bS}\tau(s|f^*)\cn_s(\rho_{i})$ is constant.
Define the cq-channel $\{W_s\}_{s\in\bS}\subset CQ([K],\kr)$ by $W_s(i):=\cn_s(\rho_i)$. We know from the results in (\cite{ahlswede-blinovsky}, see Theorem 1) and the later work (\cite{bbjn}, their Theorem 1 and 2), that $\overline C_{\mathrm{random}}(\{W_s\}_{s\in\bS})$ equals the message transmission capacity of the \emph{compound} cq-channel that is built up from $\conv(\{W_s\}_{s\in\bS})$:
\begin{align}
\overline C_{\mathrm{random}}(\{W_s\}_{s\in\bS})=\underset{p \in \mathfrak P(\mathbf{X})}{\max}\underset{W \in \conv(\{W_s\}_{s\in\bS})}{\inf} \chi(p,W).
\end{align}
Since $\conv(\{W_s\}_{s\in\bS})$ contains the constant channel $W^*$ this implies that $\{W_s\}_{s\in\bS}$ satisfies $\overline C_{\mathrm{random}}(\{W_s\}_{s\in\bS})=0$. This conclusion holds regardless of the set $\mathfrak S_K$ that was chosen, only the conditional probability distribution $\tau(s|f^*)$ changes for every possible set $\{\rho_1,\ldots,\rho_K\}$.\\
Fortunately, the whole argument can be gone through for every $\fri^{\otimes l}:=\{\cn_{s^l}\}_{s^l\in\bS^l}$, $l\in\nn$, it follows that every $\{W_{s^l}\}_{s^l\in\bS^l}$ with $W_{s^l}\in CQ([K],\hr^{\otimes l})$ defined via $W_{s^l}(i):=\cn_{s^l}(\rho_i)$ for a set $\{\rho_1,\ldots,\rho_K\}\subset\cs(\hr^{\otimes l})$ satisfies $\overline C_{\mathrm{random}}(\{W_{s^l}\}_{s^l\in\bS^l})=0$.\\
We have to show that this implies $\overline C_{\mathrm{random}}(\fri)=0$. Assuming the contrary, we know from Lemma 37 in \cite{abbn} that for every $\eps>0$ there exist $l\in\nn$, a finite set $\mathbf Z$, POVMs $\{\mathbf D^{(z)}\}_{z\in\mathbf Z}\subset\mathcal M_2(\kr^{\otimes l})$, states $\{\rho^{(z)}_1,\rho^{(z)}_2\}_{z\in\mathbf Z}\subset\cs(\hr^{\otimes l})$ and $\gamma_l\in\mathfrak P(\mathbf Z)$ such that
\begin{align}
\sum_{z\in\mathbf Z}\frac{1}{2}(\tr\{D_1^{(z)}\cn_{s^l}(\rho_1^{(z)})\}+\tr\{D_2^{(z)}\cn_{s^l}(\rho_2^{(z)}\})\gamma_l(z)\geq1-\eps\qquad\forall s^l\in\bS^l
\end{align}
and this implies
\begin{align}
\sum_{z\in\mathbf Z}\gamma_l(z)\tr\{D_i^{(z)}\cn_{s^l}(\rho_i^{(z)})\}\geq1-2\eps\qquad\forall s^l\in\bS^l,\ i=1,2.
\end{align}
Assume we have taken some large $l$ in above formula, such that $\eps<1/4$. Let, for this $l$, $\Phi:[|\bS^l|]\to\bS^l$ be a bijection and define $T:=[|\bS^l|]$. Also, let us define the channels
\begin{align}
\hat\cn_t:=\cn_{\phi(t)}.
\end{align}
In accordance with our previous notation we additionally define $\hat\cn_{t^m}:=\otimes_{i=1}^m\hat\cn_{t_i}$. Let us stop for a short \emph{remark} concerning the necessity of $\Phi$ and the corresponding set $T$: They only helps us avoiding notational skyscrapers like ${s^l}^m$, since it is our feeling that these objects (especially when they are, in addition, used as subscript of yet another symbol) would make the formulae to come very difficult to read.\\
Now, defining the arbitrarily varying channel $\mathfrak U$ via its elements $\mathfrak U_{t}:[2]\mapsto\mathfrak P([2])$, $t\in T$, where
\begin{align}
\mathfrak U_{t}(j|i):=\sum_{z\in\mathbf Z}\tr\{D_j^{(z)}\hat\cn_{t}(\rho_i^{(z)})\}\gamma_l(z)
\end{align}
we know that, since $2\eps<1/2$, that $\mathfrak U$ is not symmetrizable in the classical sense as given in \cite{ericson}. By Theorem 1 in \cite{csiszar-narayan} and the Ahlswede-Dichotomy we know that this implies that $\mathfrak U$ has a positive capacity for transmission of classical messages using deterministic coding schemes, hence there exist sequences
\begin{align}
(\{x^m_i\}_{i=1}^{M_m})_{m\in\nn},\qquad(\{E_i^m\}_{i=1}^{M_m})_{m\in\nn}
\end{align}
where $\liminf_{m\to\infty}\frac{1}{m}\log M_m=c>0$, each of the $x^m_i$ is an element of $[2]^m$ and the $E_i^m\subset[2]^m$ are mutually disjoint for constant $m$, and most importantly we have
\begin{align}
\lim_{m\to\infty}\min_{t^m\in T^m}\frac{1}{M_m}\sum_{i=1}^{M_m}\mathfrak U_{t^m}(E_i^m|x^m_i)=1.
\end{align}
Rewriting this last formula using the definition of the $\mathfrak U_{t}$ and writing $x^m_i=(x_{i,1},\ldots,x_{i,m})$ as well as $D_{y^m}^{(z^m)}:=D_{y_1}^{(z_1)}\otimes\ldots\otimes D_{y_m}^{(z_m)}$ f.a. $y^m\in[2]^m$, $z^m\in\mathbf Z^m$ we get
\begin{align}
1&=\lim_{m\to\infty}\min_{t^m\in T^m}\frac{1}{M_{m}}\sum_{i=1}^{M_{m}}\sum_{y^{m}\in E_i^m}\prod_{j=1}^m\sum_{z\in\mathbf Z}\tr\{D_{y_j}^{(z)}\hat\cn_{t_j}(\rho^{(z)}_{x_{i,j}})\}\gamma_l(z)\\
&=\lim_{m\to\infty}\min_{t^m\in T^m}\sum_{z^m\in\mathbf Z^m}\frac{1}{M_{m}}\sum_{i=1}^{M_{m}}\sum_{y^{m}\in E_i^m}\tr\{D_{y^m}^{(z^m)}\hat\cn_{t^m}(\rho_{x^m_{i}}^{(z^m)})\}\gamma_l^{\otimes m}(z^m)\\
&=\lim_{m\to\infty}\min_{t^m\in T^m}\sum_{z^m\in\mathbf Z^m}\frac{1}{M_{m}}\sum_{i=1}^{M_{m}}\tr\{D_{E_i^m}^{(z^m)}\hat\cn_{t^m}(\rho^{(z^m)}_{x^m_{i}})\}\gamma_l^{\otimes m}(z^m),
\end{align}
where we introduced the states $\rho_{x^m_{i}}^{(z^m)}:=\otimes_{j=1}^m\rho_{x^m_{i,j}}^{z_j}$ in the forelast line and the POVMs $(D_{E_i^m}^{(z^m)})_{i=1}^{M_m}$ in the last line. That this is a valid definition is seen as follows:
\begin{align}
\sum_{i=1}^{M_m}D_{E_i^m}^{(z^m)}&=\sum_{i=1}^{M_m}\sum_{y^m\in E_{i^m}}D_{y^m}^{(z^m)}\\
&\leq\sum_{y^m\in[2]^m}D_{y^m}^{(z^m)}\\
&\leq\underbrace{\eins_{\hr^{\otimes l}}\otimes\ldots\otimes\eins_{\hr^{\otimes l}}}_{m-\mathrm{times}}.
\end{align}
This shows that with
\begin{align}
\cP^{z^m}_m:[M_m]\rightarrow\cs(\hr^{\otimes ml}),\qquad i\mapsto \rho^{(z^m)}_{x^m_i},
\end{align}
we have a suitable sequence of random codes with a positive rate for the $cq$ channel $\{W_{s^l}\}_{s^l\in\bS^l}\subset CQ([K]\times\mathbf Z,\hr^{\otimes l})$ arising via $W_{s^l}((i,z)):=\cn_{s^l}(\rho_i^{(z)})$.\\
But this is in contradiction to what we have already proven to be true, hence $\overline C_{\mathrm{random}}(\fri)=0$.\\
Now let $r$ be arbitrary. Then the above argument, applied to $\fri^r:=\{\cn_{s^r}\}_{s^r\in\bS^r}$ shows that $\overline C_{\mathrm{random}}(\fri^r)=0$ has to hold. But $C_{\mathrm{random}}(\fri^r)=r\cdot C_{\mathrm{random}}(\fri)$, hence $\overline C_{\mathrm{random}}(\fri)=0$.
\end{proof}
\begin{proof}[Proof of Theorem \ref{thm-2}] Let $I(X,Y)>0$ and $r\in\nn$. Then by Theorem \ref{thm-1} we have that $\overline C(\fri,r,(X,Y))=0$ implies $\overline C_{\mathrm{random}}(\fri)=0$. Let
$\overline C(\fri,r,(X,Y))>0$. Loosely speaking, we will use this assumption to transmit a small number of messages first, using the correlation that is present in the source $(X,Y)$. The messages sent that way are highly correlated with the decoder (meaning that, when the sender sends message '$i$' then the receiver will, with high probability, detect this) and can therefore be used to operate a random code for message transmission that uses only a small amount of common randomness to correlate the en- and decoding, but achieves a high message transmission rate. During this second part of the transmission, sender and receiver can simply ignore the correlation that is emitted by $(X,Y)$.\\
While the first part of the protocol (running on roughly $\log l$ channel uses) only establishes the common randomness, the second one (using the channel $(l-\log l)$ times) is used to achieve a high throughput of messages - actually, the message transmission rate will be arbitrarily close to optimal.\\
To get started, we define the (possibly very small) positive number $c:=\overline C(\fri,r,(X,Y))/2$. From the definition of capacity it is then clear that there exists a sequence $(\mathfrak C_l)_{l\in\nn}$ of $((X,Y),r)$-correlated codes for message transmission over $\fri$ such that
$N(l):=\lfloor2^{c\cdot l}\rfloor$ messages are transmitted (at blocklength $l$, and for all $l\in\nn$) and this transmission is asymptotically perfect with respect to the average error criterion:
\begin{align}
\sum_{i=1}^{N(l)}\frac{1}{N(l)}\sum_{x^{n(l)},y^{n(l)}}p^{\otimes n(l)}(x^{n(l)},y^{n(l)})\tr\{D_{y^{n(l)},i}\cn_{s^{l}}\circ\cP_{x^{n(l)}}(i)\}\geq1-\delta_l,\qquad \delta_l\searrow0.
\end{align}
Take a sequence of random codes $(\{(M_l,\mathcal D_i,\cP_i)\}_{i=1}^{l^2})_{l\in\nn}$ for message transmission over $\fri$ such that
\begin{align}
\min_{s^l\in\bS^l}\sum_{i=1}^{l^2}\frac{1}{l^2}\frac{1}{M_l}\sum_{j=1}^{M_l}\tr\{D_{i,j}\cn_{s^l}\circ\cP^{l}_i(j)\}\geq1-\eps_l,\qquad\eps_l\searrow0,\\
\liminf_{l\to\infty}\frac{1}{l}\log M_l\geq\overline C_{\mathrm{random}}(\fri)-\eta,\qquad \eta>0\ \mathrm{arbitrary}.
\end{align}
Such codes are guaranteed to exist by Lemma 9 and Lemma 10 in \cite{bbjn}. Let $m(l):=\lfloor\frac{2}{c}\log l\rfloor$, $\overline m(l):=\max\{\lfloor l/r\rfloor \ :\ \lfloor l/r\rfloor\leq m(l)\}$, $c(l):=l-m(l)$. Consider the code defined by
\begin{equation}
\cP^l_{x^{n(l)}}(j):=\sum_{i=1}^{l^2}\frac{1}{l^2}\cP^{c(l)}_i(j)\otimes\cP_{x^{\overline m(l)}}(i),\qquad D_{y^{n(l)},j}:=\sum_{i=1}^{l^2}D_{i,j}\otimes D_{y^{\overline m(l)},i},
\end{equation}
where for $x^{n(l)}=(x_1,\ldots,x_{n(l)})$ we define $x^{\overline m(l)}:=(x_1,\ldots,x_{\overline m(l)})$. Then
\begin{align}
&\sum_{x^{n(l)}}\sum_{y^{n(l)}}p^{\otimes n(l)}(x^{n(l)},y^{n(l)})\sum_{j=1}^{M_{c(l)}}\tr\{D_{y^{n(l)},j}\cn_{s^l}\circ\cP_{x^{n(l)}}(j)\}\\
&\geq\sum_{x^{n(l)},y^{n(l)}}p^{\otimes n(l)}(x^{n(l)},y^{n(l)})\sum_{j=1}^{M_{c(l)}}\sum_{i=1}^{l^2}\frac{1}{l^2}\tr\{(D_{i,j}\otimes D_{y^{\overline m(l)},i})\cn_{s^l}\circ(\cP^{c(l)}_i(j)\otimes\cP_{x^{\overline m(l)}}(i))\}\\
&=\sum_{i=1}^{l^2}\frac{1}{l^2}\sum_{x^{\overline m(l)},y^{\overline m(l)}}p^{\otimes \overline m(l)}(x^{\overline m(l)},y^{\overline m(l)})\tr\{D_{y^{\overline m(l)},i}\cn_{s^{ m(l)}}\circ\cP_{x^{\overline m(l)}}(i)\}\sum_{j=1}^{M_{c(l)}}\tr\{D_{i,j}\cn_{s^{c(l)}}\circ\cP^{c(l)}_i(j)\}\\
&\geq M_l(1-\eps_{c(l)}-\delta_{m(l)}),
\end{align}
by the Innerproduct Lemma in \cite{ahlswede-elimination} (which is, to be precise, stated there with $\eps_{c(l)}=\delta_{m(l)}=\eps$, but holds in our slightly more general
case as well). Additionally,
\begin{align}
\liminf_{l\to\infty}\frac{1}{l}\log M_{c(l)}&=\liminf_{l\to\infty}\frac{1}{m(l)+c(l)}\log M_{c(l)}\geq \overline C_{\mathrm{random}}(\fri)-\eta.
\end{align}
\end{proof}
\end{section}
\begin{section}{\label{sec:correlation-and-entanglement}Impact of Correlation on Entanglement Transmission over an AVQC}
We will now give the proof of Theorem \ref{thm:entanglement-transmission-over-avqc}, showing that for a finite AVQC $\fri$, the use of correlated codes is sufficient to guarantee transmission of entanglement at a rate of $\A_{\mathrm{random}}(\fri)$. This shows that a very restricted class of random codes, namely $((X,Y),r)$ correlated codes with arbitrary $r\in\nn$ and the only requirement being $I(X,Y)>0$ already achieve the full random entanglement transmission capacity. In terms of the resources that are needed in order to transmit at the random entanglement transmission capacity, this is a huge benefit.\\
The idea of the proof is the same as the one for the proof of Theorem \ref{thm-2}, with the only difference being that entanglement is transmitted in the second part of the protocol instead of messages. Despite the similarity of the two proofs, we still give the full proof of Theorem \ref{thm:entanglement-transmission-over-avqc} - for readers convenience.
\begin{proof}[Proof of Theorem \ref{thm:entanglement-transmission-over-avqc}]
By Theorem \ref{thm-2}, $\overline C(\fri,r,(X,Y))>0$ for every $r\in\nn$. Again we set $c:=\overline C(\fri,r,(X,Y))/2$. Thus, for every such $r$ there exists a sequence $(\mathfrak C_l)_{l\in\nn}$ of $((X,Y),r)$-correlated codes for message transmission over $\fri$ such that
$N(l):=\lfloor2^{c\cdot l}\rfloor$ messages are transmitted (at blocklength $l$, and for all $l\in\nn$) and the transmission of messages is asymptotically perfect with respect to the average error criterion:
\begin{align}
\sum_{i=1}^{N(l)}\frac{1}{N(l)}\sum_{x^{n(l)},y^{n(l)}}p^{\otimes n(l)}(x^{n(l)},y^{n(l)})\tr\{D_{y^{n(l)},i}\cn_{s^{l}}\circ\cP_{x^{n(l)}}(i)\}\geq1-\delta_l,\qquad \delta_l\searrow0.
\end{align}
Take a sequence of random codes $(\{\fr_l,\hat\crr^l_i,\hat\cP^l_i\}_{i=1}^{l^2})_{l\in\nn}$ for message transmission over $\fri$ such that
\begin{align}
\min_{s^l\in\bS^l}\sum_{i=1}^{l^2}\frac{1}{l^2}F_e(\pi_{\fr_l},\hat\crr^l\circ\cn_{s^l}\circ\hat\cP^l)\geq1-\eps_l,\qquad\eps_l\searrow0,\\
\liminf_{l\to\infty}\frac{1}{l}\log \dim(\fr_l)\geq\A_{\mathrm{random}}(\fri)-\eta,\qquad \eta>0\ \mathrm{arbitrary}.
\end{align}
Existence of such a sequence is guaranteed by Theorem 32 and Lemma 37 in \cite{abbn}. Let $m(l):=\frac{2}{c}\log l$, $\overline m(l):=\max\{\lfloor l/r\rfloor \ :\ \lfloor l/r\rfloor\leq m(l)\}$, $c(l):=l-m(l)$. Consider the code defined by
\begin{align}
\cP_{x^{n(l)}}^l(a)&:=\sum_{i=1}^{l^2}\frac{1}{l^2}\cdot\hat\cP^{c(l)}_i(a)\otimes\cP_{x^{\overline m(l)}}(i)\qquad\forall a\in\mathcal B(\fr_l),\\
\crr^l_{y^{n(l)}}(a\otimes b)&:=\sum_{i=1}^{m(l)}\tr\{D_{y^{\overline m(l)},i}b\}\cdot\crr^{c(l)}_i(a)\qquad\forall a\in\B(\kr^{\otimes c(l)}),\ b\in\B(\kr^{\otimes m(l)}).
\end{align}
Then for every $s^l=(s^{m(l)},s^{c(l)})$ it holds
\begin{align}
&\sum_{x^{n(l)}}\sum_{y^{n(l)}}p^{\otimes n(l)}(x^{n(l)},y^{n(l)})F_e(\pi_{\fr_l},\crr^l_{y^{n(l)}}\circ\cn_{s^l}\circ\cP^l_{x^{n(l)}})\\
&\geq\sum_{i=1}^{l^2}\frac{1}{l^2}\sum_{x^{\overline m(l)},y^{\overline m(l)}}p^{\otimes \overline m(l)}(x^{\overline m(l)},y^{\overline m(l)})\tr\{\cn_{s^{m(l)}}(\cP_{x^{\overline m(l)}}(i))D_{y^{\overline m(l)},i}\}\cdot F_e(\pi_{\fr_{c(l)}},\crr^{c(l)}_i\circ\cn_{s^{c(l)}}\circ\cP^{c(l)}_i)\\
&\geq 1-\eps_{c(l)}-\delta_{m(l)},
\end{align}
 where we use the Innerproduct Lemma in \cite{ahlswede-elimination} the second time. Additionally,
\begin{align}
\liminf_{l\to\infty}\frac{1}{l}\log \dim(\fr_{c(l)})&=\liminf_{l\to\infty}\frac{1}{m(l)+c(l)}\log \dim(\fr_{c(l)})\geq \overline \A_{\mathrm{random}}(\fri)-\eta.
\end{align}
\end{proof}
\end{section}
\begin{section}{\label{sec:equivalence}Equivalence of maximal- and average error criterion}
As mentioned in the introduction, the asymptotic equivalence of the two error criteria already exists in the classical case, when one allows for randomized encoding. While it is no wonder that a code that has a good performance with respect to the maximal error criterion has at least the same performance with respect to the average error criterion, the reverse statement is nontrivial.\\
A simple expurgation of a small number of 'bad' codewords, as is done for a single memoryless channel, would not suffice in the setting considered here - for each possible choice of channel sequence that the jammer might make, one would have to expurgate possibly \emph{different} subsets codewords. Taking into account that the jammer has an exponentially large number of choices it is clear that this will not yield the desired result.\\
Therefore, a different strategy is chosen:
\begin{proof}[Proof of Theorem \ref{thm:equivalence-of-average-and-maximal-error}]
In the first part of the proof we show that $\overline C_{\mathrm{det}}(\fri)>0$ if and only if $C_{\mathrm{det}}(\fri)>0$. This part of the proof uses the symmetrizability conditions of \cite{abbn}. Then, we show that $C_{\mathrm{det}}(\fri)>0$ implies $C_{\mathrm{det}}(\fri)=C_{\mathrm{random}}(\fri)$: We first use Ahlswede's random code reduction (this proves that not too much randomness is needed in order to have asymptotically perfect transmission of messages with respect to the maximal error criterion and at a rate close to $C_{\mathrm{random}}(\fri)$), followed by the 'usual' trick: First, use the fact that $C_{\mathrm{det}}(\fri)>0$ to establish some small amount of common randomness between sender and receiver, then operate a random code on top of that. Finally, we show that $\overline C_{\mathrm{random}}(\fri)=C_{\mathrm{random}}(\fri)$ - given an arbitrary code for transmission of messages that works well under average error criterion, the very same code, randomized jointly over all permutations of the en- and decoding operations, works well for the maximal error criterion. Let us start now with the details.
\\\\
It is clear that $C_{\mathrm{det}}(\fri)>0$ implies $\overline C_{\mathrm{det}}(\fri)>0$. On the other hand, if $\overline C_{\mathrm{det}}(\fri)>0$ then for every $\eps>0$ there is $l\in\nn$ and $\rho_1,\rho_2\in\cs(\hr^{\otimes l})$ as well as a POVM consisting of $D_1,D_2$ such that
\begin{align}
\frac{1}{2}(\tr\{D_1\cn_{s^l}(\rho_1)\}+\tr\{D_2\cn_{s^l}(\rho_2)\})\geq1-\eps\qquad\forall\ s^l\in\bS^l.
\end{align}
As usual, this implies
\begin{align}\label{eqn-101}
\tr\{D_i\cn_{s^l}(\rho_i)\}\geq1-2\eps\qquad\forall\ s^l\in\bS^l,\ i=1,2.
\end{align}
Assume $\fri$ has $C_{\mathrm{det}}(\fri)=0$, then by Theorem 42 in \cite{abbn} we have that there exist $p,q\in\mathfrak P(\bS^l)$ such that with $\cn_p:=\sum_{s^l\in\bS^l}p(s^l)\cn_{s^l}$ and $\cn_q:=\sum_{s^l\in\bS^l}q(s^l)\cn_{s^l}$ we have
\begin{align}\label{eqn-102}
\sum_{s^l\in\bS^l}p(s^l)\cn_{s^l}(\rho_1)=\sum_{s^l\in\bS^l}q(s^l)\cn_{s^l}(\rho_2),
\end{align}
while from equation (\ref{eqn-101}) we know that
\begin{align}
\tr\{D_i\cn_{q}(\rho_i)\}\geq1-2\eps,\qquad\tr\{D_i\cn_{p}(\rho_i)\}\geq1-2\eps\qquad, i=1,2.
\end{align}
Combining this with equation (\ref{eqn-102}) gives the contradiction
\begin{align}
1-2\eps&\leq\tr\{D_2\cn_q(\rho_2)\}\\
&=\tr\{D_2\cn_p(\rho_1)\}\\
&=1-\tr\{D_1\cn_p(\rho_1)\}\\
&\leq2\eps,
\end{align}
and this is a clear contradiction for $\eps<1/4$. But $\eps>0$ was arbitrary, proving our claim.
\\\\
The proof that $C_{\mathrm{det}}(\fri)>0$ implies $C_{\mathrm{det}}(\fri)=C_{\mathrm{random}}(\fri)$ is again straight along the lines of the proof in the classical case, which was done by Ahlswede in \cite{ahlswede-elimination}. For readers convenience, we give a proof that is a slight modification of our proof of Lemma 10 in \cite{bbjn}, which in turn is just a reformulation of Ahlswedes original proof to the quantum setting.
\begin{lemma}[Random Code Reduction]\label{random-code-reduction}
Let $\A=\{A_s\}_{s\in\mathbf S}$ be an AVQC, $(\mu_l)_{l\in\nn}$ a sequence of random codes for message transmission over $\fri$ such that
\begin{equation}\label{eq:random-code-reduction}
\min_{s^l\in\bS^l}\min_{i\in[M_l]}\int\tr\{\cn_{s^l}(\rho_i)D_i\}d\mu_l((x_i,D_i)_{i=1}^{M_l})\ge
1-\eps_l,
\end{equation}
where $\eps_l\searrow0$ and $\liminf_{l\to\infty}\frac{1}{l}\log(M_l)=R>0$. For every $\eps>0$ there is $L\in\nn$ such that there exist $L^2$ $(L,M_L)$-deterministic codes
$(\rho_{1,j},\ldots,\rho_{M_L,j},D_{1,j},\ldots,D_{M_l,j})$ ($1\leq j\leq L^2$) for $\fri$  such that
\begin{equation}\label{eq:random-code-reduction-a}
\min_{i\in[M_L]}\frac{1}{L^2}\sum_{j=1}^{L^2}\tr\{\cn_{s^L}(\rho_{i,j})D_{i,j}\}\geq1-\eps \qquad  \forall s^L\in\mathbf S^L.
\end{equation}
\end{lemma}
\begin{proof}
Let $l\in\nn$ be arbitrarily large but fixed for the moment. For a fixed $K\in\nn$, consider $K$ independent random variables $\Lambda_i$ with values in $((\mathbf X^l)^{M_l}\times\M_{M_l}(\hr^{\otimes l}))$
which are distributed according to
$\mu_l$.\\
Define, for each $s^l\in\bS^l$ and $i\in[M_l]$, the function $p^i_{s^l}:((\mathbf X^l)^{M_l}\times\M_{M_l}(\hr^{\otimes l}))\rightarrow[0,1]$,\\
$(\rho_1,\ldots,\rho_{M_l},D_1,\ldots,D_{M_l})\mapsto\tr\{\cn_{s^l}(x^n_i)D_i\}$.\\
We get, by application of Markovs inequality, for every $s^l\in\bS^l$ and $i\in[M_l]$:
\begin{eqnarray}
 \mathbb P(1-\frac{1}{K}\sum_{j=1}^{K}p^i_{s^l}(\Lambda_j)\geq\eps)&=& \mathbb P(2^{ K- \sum_{j=1}^Kp^i_{s^l}(\Lambda_j)}\geq2^{K\eps})\\
&\leq&2^{-K\eps}\mathbb E(2^{ (K-\sum_{j=1}^Kp^i_{s^l}(\Lambda_j))}).
\end{eqnarray}
The $\Lambda_i$ are independent and it holds $2^{t}\leq1+t$ for every $t\in[0,1]$ as well as $\log(1+\eps_l)\leq2\eps_l$ and so we get
\begin{align}
 \mathbb P(1-\frac{1}{K}\sum_{j=1}^{K}p^i_{s^l}(\Lambda_j)\geq\eps)&\leq 2^{-K\eps}\mathbb E(2^{K-\sum_{j=1}^Kp^i_{s^l}(\Lambda_j)})\\
&=2^{-K\eps}\mathbb E(2^{\sum_{j=1}^K(1-p^i_{s^l}(\Lambda_j))})\\
&=2^{-K\eps}\mathbb E(2^{(1-p^i_{s^l}(\Lambda_1))})^K\\
&\leq2^{-K\eps}\mathbb E(1+(1-p^i_{s^l}(\Lambda_1)))^K\\
&\leq2^{-K\eps}\mathbb E(1+\eps_l)^K\\
&\leq2^{-K\eps}2^{K2\eps_l}\\
&=2^{-K(\eps-2\eps_l)}.
\end{align}
Therefore, and since we can w.l.o.g assume that not only $\liminf_{l\to\infty}\frac{1}{l}\log M_l=R$, but even $\lim_{l\to\infty}\frac{1}{l}\log M_l=R$,
\begin{align}
  \mathbb P(\frac{1}{K}\sum_{j=1}^{K}p^i_{s^l}(\Lambda_j)\geq1-\eps\ \forall s^l\in\bS^l,\ i\in[M_l])&\geq1-|M_l|\cdot|\bS|^l\cdot2^{-K(\eps-2\eps_l)}\\
  &\geq1-2^{l(R+\eps+\log|\bS|)}\cdot2^{-K(\eps-2\eps_l)}.
  \end{align}
By assumption, $\eps_l\searrow0$ and a comparison of the order of the exponentials in above estimate yields the existence of $L\in\nn$ satisfying
\begin{align}
\frac{1}{L^2}\sum_{i=1}^{L^2}\tr(\cn_{s^L}(\rho_i)D_i)\geq1-\eps\qquad\forall i\in[M_L],\ s^L\in\bS^L.
\end{align}
\end{proof}
Now, using Lemma \ref{random-code-reduction} and carrying out the derandomization procedure as in the proof of Theorem 2 in \cite{bbjn}, but with maximal error probability criterion instead of average error criterion, one verifies that $C_{\mathrm{det}}(\fri)=C_{\mathrm{random}}(\fri)$, if the l.h.s. is strictly larger than zero.\\
Since the inequality $C_{\mathrm{random}}(\fri)\leq \overline C_{\mathrm{random}}(\fri)$ is trivially true, all that is left to do is to show the reverse inequality. To this end, let $(\mu_l)_{l\in\nn}$ be a sequence of random codes such that
\begin{align}
\min_{s^l\in\bS^l}\int \frac{1}{M_l}\sum_{i=1}^{M_l}\tr\{\cn_{s^l}(\cP(i))D_i\}d\mu_l(\cP,\mathbf D)\ge
1-\eps_l\qquad\forall l\in\nn,
\end{align}
where $\eps_l\searrow0$ and $\liminf_{l\to\infty}\frac{1}{l}\log M_l\geq\overline C_{\mathrm{random}}(\fri)-\eta$ for some $\eta>0$ that we may choose arbitrarily small.\\
We define the new sequence of random codes $(\hat\mu_l)_{l\in\nn}$ by
\begin{align}
\hat\mu_l(A):=\sum_{\tau\in\mathrm{Perm}([M_l])}\mu_l(\tau(A))\frac{1}{M_l!},\qquad A\in \Sigma_l,
\end{align}
where $\mathrm{Perm}([M_l])$ is the set of permutations on $M_l$ and the for a pair $(\cP,\mathbf D)\in CQ([M_l],\hr^{\otimes l})\times\M_{M_l}(\kr^{\otimes l})$ we define $\tau(\cP,\mathbf D):=(\cP\circ\tau,(D_{\tau(1)},\ldots,D_{\tau(M_l)}))$. Then for every $l\in\nn$, $s^l\in\bS^l$ and $i\in[M_l]$ we get
\begin{align}
\int\tr\{\cn_{s^l}(\cP(i))D_i\}d\hat\mu_l(\cP,\mathbf D)&=\sum_{\tau\in\mathrm{Perm}([M_l])}\int\tr\{\cn_{s^l}(\cP(\tau(i)))D_\tau(i)\})\frac{1}{M_l!}d\mu_l(\cP,\mathbf D)\\
&=\sum_{i=1}^{M_l}\int\tr(\cn_{s^l}(\cP(i))D_i)\frac{1}{M_l}d\mu_l(\cP,\mathbf D)\\
&\geq1-\eps_l.
\end{align}
\end{proof}
\end{section}
\begin{section}{\label{sec:symmetrizability}The Simplified Symmetrizability Condition}
The following theorem is the core for our analysis of the symmetrizability condition Definition 39 given in \cite{abbn} and the proof of Theorem \ref{thm:geometric-variant-of-symmetrizability}.
\begin{theorem}\label{theorem-elementary-convex-reduction}
Let $\{N_s\}_{s\in\bS}$ be a finite set of linear maps from $\mathbb C^n$ to $\mathbb C^m$. If, for $\mathfrak S=\{t_1,\ldots,t_K\}\subset\mathbb C^n$ there exist probability distributions $p_1\ldots,p_K\subset\mathfrak P(\bS)$ such that
\begin{equation}
 \sum_{s\in\bS}p_i(s)N_s(t_j)=\sum_{s\in\bS}p_j(s)N_s(t_i)\qquad\forall i,j\in[K],
\end{equation}
then for every $\{t_{K+1},\ldots,t_N\}\subset \conv(\mathfrak S)$ ($N\geq K$) there exist $p_{K+1},\ldots,p_N\in\mathfrak P(\bS)$ such that
\begin{equation}
 \sum_{s\in\bS}p_i(s)N_s(t_j)=\sum_{s\in\bS}p_j(s)N_s(t_i)\qquad\forall i,j\in[N].
\end{equation}
\end{theorem}
\begin{proof}
Take a classical channel $r:[N]\rightarrow\mathfrak P[N]$ (for each $i\in[N]$, $r(\cdot|i)\in\mathfrak P([N])$ or, equivalently, $r$ is a conditional probability distribution) such that
\begin{equation}
 t_i=\sum_{j=1}^Nr(j|i)t_j\qquad \forall i\in[N].
\end{equation}
We may choose $r$ such that $r(j|i)=\delta(i,j)$ for all $j\in[N]$ and $i\in[K]$. Define $p_{K+1},\ldots,p_N$ by
\begin{equation}
p_i:=\sum_{j=1}^Kr(j|i)p_j,
\end{equation}
then we get that
\begin{align}
\sum_{s\in\bS}p_i(s)N_s(t_j)&=\sum_{s\in\bS}\sum_{m=1}^Kr(m|j)p_i(s)N_s(t_m)\\
&=\sum_{s\in\bS}\sum_{n=1}^K\sum_{m=1}^Kr(m|j)r(n|i)p_n(s)N_s(t_m)\\
&=\sum_{s\in\bS}\sum_{n=1}^K\sum_{m=1}^Kr(m|j)r(n|i)p_m(s)N_s(t_n)\\
&=\sum_{s\in\bS}p_j(s)N_s(t_i).
\end{align}
\end{proof}
This immediately implies the following.
\begin{theorem}\label{theorem-convex-reduction}
If, for $\mathfrak S=\{\rho_1,\ldots,\rho_K\}\subset\cs(\hr)$ there exist probability distributions $p_1\ldots,p_K\subset\mathfrak P(\bS)$ such that
\begin{equation}
 \sum_{s\in\bS}p_i(s)\cn_s(\rho_j)=\sum_{s\in\bS}p_j(s)\cn_s(\rho_i)\qquad\forall i,j\in[K],
\end{equation}
then for every $\{\rho_{K+1},\ldots,\rho_N\}\subset \conv(\mathfrak S)$ ($N\geq K$) there exist $p_{K+1},\ldots,p_N\in\mathfrak P(\bS)$ such that
\begin{equation}
 \sum_{s\in\bS}p_i(s)\cn_s(\rho_j)=\sum_{s\in\bS}p_j(s)\cn_s(\rho_i)\qquad\forall i,j\in[N].
\end{equation}
\end{theorem}
Yet another formulation is the corollary
\begin{corollary}\label{corollary-1}
If, for a given set $\mathfrak S=\{\rho_1,\ldots,\rho_K\}$, there exists \emph{no} set $\{p_1,\ldots,p_K\}\subset\mathfrak P(\bS)$ such that
\begin{align}
  \sum_{s\in\bS}p_i(s)\cn_s(\rho_j)=\sum_{s\in\bS}p_j(s)\cn_s(\rho_i)\qquad\forall i,j\in[K],
\end{align}
then also for every set $\mathfrak S'=\{\sigma_1,\ldots,\sigma_{K'}\}\subset\cs(\hr)$ with $\conv(\mathfrak S)\subset\conv(\mathfrak S')$ there can be no set
$\{p_1,\ldots,p_{K'}\}\subset\mathfrak P(\bS)$ such that
\begin{align}
  \sum_{s\in\bS}p_i(s)\cn_s(\rho_j)=\sum_{s\in\bS}p_j(s)\cn_s(\rho_i)\qquad\forall i,j\in[K'].
\end{align}
\end{corollary}

\begin{theorem}\label{thm:pure-state-symmetrizability}
Let $\fri=\{\cn_s\}_{s\in\bS}\subset\mathcal C(\hr,\kr)$ be an AVQC. $\fri$ is symmetrizable if and only if for every $l\in\nn$ and every set $\{\psi_1,\ldots,\psi_K\}\subset\cs(\hr^{\otimes l})$ of pure states there is a set $\{p_1,\ldots,p_K\}\subset\mathfrak P(\bS^l)$ such that
\begin{align}
\sum_{s^l\in\bS^l}p_j(s^l)\cn_{s^l}(\psi_i)=\sum_{s^l\in\bS^l}p_i(s^l)\cn_{s^l}(\psi_j)\qquad\forall i,j\in[K].
\end{align}
\end{theorem}
\begin{remark}
This result shows that randomized encoding (e.g. the use of mixed signal states instead of pure ones), in analogy to the error probabilities defined in \cite{ahlswede-elimination}, cannot lead to a nonzero message transmission capacity if the message transmission capacity using pure signal states is zero. This is in analogy to Theorem 3 a) in \cite{ahlswede-elimination}.
\end{remark}
\begin{proof}
If $\fri$ is symmetrizable then to every set $\{\rho_1,\ldots,\rho_K\}\subset\cs(\hr^{\otimes l})$, $l\in\nn$ arbitrary, there is a set $\{p_1,\ldots,p_K\}\subset\mathfrak P(\bS^l)$ such that
\begin{align}
\sum_{s^l\in\bS^l}p_j(s^l)\cn_{s^l}(\rho_i)=\sum_{s^l\in\bS^l}p_i(s^l)\cn_{s^l}(\rho_j)\qquad\forall i,j\in[K].
\end{align}
Clearly then, this holds especially for sets $\mathfrak S=\{\psi_1,\ldots,\psi_K\}$ of pure states.\\
If $\fri$ is not symmetrizable, then there is a set $\{\rho_1,\ldots,\rho_K\}$ such that no set $\{p_1,\ldots,p_K\}\subset\mathfrak P(\bS^l)$ satisfies
\begin{align}
\sum_{s^l\in\bS^l}p_j(s^l)\cn_{s^l}(\rho_i)=\sum_{s^l\in\bS^l}p_i(s^l)\cn_{s^l}(\rho_j)\qquad\forall i,j\in[K].
\end{align}
Clearly, there is a set $\mathfrak S'=\{\psi_1,\ldots,\psi_K\}$ of pure states such that $\conv(\mathfrak S)\subset\conv(\mathfrak S')$ holds and, by Theorem \ref{theorem-convex-reduction} (Corollary \ref{corollary-1}), this implies that there can be no set $\{q_1,\ldots,q_K\}\subset\mathfrak P(\bS^l)$ such that
\begin{align}
\sum_{s^l\in\bS^l}q_j(s^l)\cn_{s^l}(\psi_i)=\sum_{s^l\in\bS^l}q_i(s^l)\cn_{s^l}(\psi_j)\qquad\forall i,j\in[K].
\end{align}
\end{proof}
\begin{proof}[Proof of Theorem \ref{thm:geometric-variant-of-symmetrizability}] The proof can be carried out using the same strategy as in the proof of Theorem \ref{thm:pure-state-symmetrizability}. A short argument goes as follows. If $\fri$ is $l$-symmetrizable according to Theorem \ref{thm:geometric-variant-of-symmetrizability}, then by Theorem \ref{theorem-elementary-convex-reduction} it is clearly symmetrizable in the sense of \cite{abbn}. If, on the other hand, it is not symmetrizable in the sense of \cite{abbn}, then it can, again by Theorem \ref{theorem-elementary-convex-reduction} not be symmetrizable in the sense of Theorem \ref{thm:geometric-variant-of-symmetrizability}.
\end{proof}
\end{section}
\ \\\\
\emph{Final Remarks.} After finishing this paper, we found explicit examples showing that our second conjecture (there exist AVQCs $\fri$ such that $\overline C_{\mathrm{det}}(\fri)=0$ but still $\overline C_{\mathrm{random}}(\fri)>0$) is true. These results will be presented in a forthcoming paper.\\
This further underlines the importance of the symmetrizability conditions that were found in \cite{abbn} as a means of distinguishing two important classes of AVQCs.\\
In view of possible applications of our results to cryptographic scenarios it should be noted that non-symmetrizability is a necessary (and, by the preceeding, also nontrivial) criterion for the ability to share or generate a secret key over a quantum channel that is under partial control of a malicious party that can do eavesdropping, but is also able to carry out active manipulations of the channel.
\\\\
\emph{Acknowledgements.} We thank Andreas Winter for stimulating discussions and remarks on the topic.
This work was supported by the DFG via grant BO 1734/20-1 (H.B.) and by the BMBF via grant 01BQ1050 (H.B., J.N.).

\end{document}